\documentclass[journal]{IEEEtran}
\usepackage{hyperref}

\usepackage{cite}

\ifCLASSINFOpdf
  \usepackage[pdftex]{graphicx}
\else
  \usepackage[dvips]{graphicx}
\fi
\usepackage{amsmath}
\usepackage{algorithmic}

\usepackage{array}

\usepackage[caption=false,font=normalsize,labelfont=sf,textfont=sf]{subfig}
\let\MYorigsubfloat\subfloat
\renewcommand{\subfloat}[2][\relax]{\MYorigsubfloat[]{#2}}
\usepackage{url}

\usepackage{amssymb,cases}
\usepackage{algorithm}
\usepackage{nccmath}

\usepackage{color}

\newtheorem{theorem}{Theorem}
\newtheorem{lemma}{Lemma}
\newtheorem{proposition}{Proposition}
\newenvironment{proof}{\quad \emph{Proof:}}{\hfill{$\square$}}

\def\d{\ldots}
\def\cK{\mathcal{K}}
\def\cM{\mathcal{M}}

\def\cCN{\mathcal{CN}}

\def\st{\mathrm{s.t.}}

\def\tr{\mathrm{tr}}
\def\rk{\mathrm{rank}}
\def\diag{\mathrm{diag}}

\renewcommand{\v}[1]{\boldsymbol{\mathbf{#1}}}
\newcommand{\m}[1]{\boldsymbol{\mathbf{#1}}}

\def\hermitian{\dagger}
\def\transpose{\mathsf{T}}

\newcommand{\algref}[1]{Algorithm~\ref{#1}}

\makeatletter
\newcommand{\pushright}[1]{\ifmeasuring@#1\else\omit\hfill$\displaystyle#1$\fi\ignorespaces}
\newcommand{\pushleft}[1]{\ifmeasuring@#1\else\omit$\displaystyle#1$\hfill\fi\ignorespaces}
\makeatother

\newcommand{\Lm}[1][ ]{\left\{\m{\Lambda}_m^{#1}\right\}}

\newcommand{\Vk}[1][ ]{\{\m V_k^{#1}\}}

\newcommand\eqstd{\eqref{equ:dual_const2}--\eqref{equ:dual_const1}}

\newcommand\eqsta{\eqref{equ:dual_const2}--\eqref{equ:pri_var2}}

\def\fk{~\forall ~k \in \cK}
\def\fm{~\forall ~m \in \cM}
\def\fmm{~\forall ~m \in \tilde{\cM}}

\title{QoS-Based Beamforming and Compression Design for Cooperative Cellular Networks via Lagrangian Duality}

\author{\IEEEauthorblockN{Xilai~Fan, Ya-Feng~Liu, Liang~Liu, and Tsung-Hui Chang}
	\thanks{Part of this work has been presented at the IEEE International Conference on Acoustics, Speech, and Signal Processing (ICASSP) 2022 \cite{fan2022EfficientlyGloballySolving}. 
	The work of X. Fan and Y.-F. Liu was supported in part by the National Natural Science Foundation of China (NSFC) under Grant 12371314 and Grant 12021001. 
	The work of L. Liu was supported in part by CRF Young Collaborative Research Grant from the Hong Kong Research Grants Council under Grant PolyU C5002-23Y; in part by the General Research Fund from the Hong Kong Research Grants Council under Grant 15213322. 
	The work of Tsung-Hui Chang was supported in part by Shenzhen Science and Technology Program under Grant RCJC20210609104448114 and Grant ZDSYS20230626091302006 and in part by Guangdong Provincial Key Laboratory of Big Data Computing. \textit{(Corresponding author: Ya-Feng Liu)}
	
	X. Fan is with the State Key Laboratory of Scientific and Engineering Computing, Institute of Computational Mathematics and Scientific/Engineering Computing, Academy of Mathematics and Systems Science, Chinese Academy of Sciences, Beijing 100190, China (e-mail: fanxilai@lsec.cc.ac.cn). 
	
	Y.-F. Liu is with the Ministry of Education Key Laboratory of Mathematics
	and Information Networks, School of Mathematical Sciences, Beijing University
	of Posts and Telecommunications, Beijing 100876, China (email: yafengliu@bupt.edu.cn).
	
	L. Liu is with the Department of Electronic and Information Engineering, The Hong Kong Polytechnic University, Hong Kong SAR, China (e-mail: liang-eie.liu@polyu.edu.hk). 
	
	T.-H. Chang is with the School of Science and Engineering, The Chinese University of Hong Kong, Shenzhen, Shenzhen, China, and Shenzhen Research Institute of Big Data (e-mail: tsunghui.chang@ieee.org). 
}}

\begin{document}

\maketitle

\begin{abstract}
	This paper considers the quality-of-service (QoS)-based joint beamforming and compression design problem in the downlink cooperative cellular network, 
	where multiple relay-like base stations (BSs), connected to the central processor via rate-limited fronthaul links, cooperatively transmit messages to the users. 
	The problem of interest is formulated as the minimization of the total transmit power of the BSs, subject to all users' signal-to-interference-plus-noise ratio (SINR) constraints and all BSs' fronthaul rate constraints.  
	In this paper, we first show that there is no duality gap between the considered joint optimization problem and its Lagrangian dual by showing the tightness of its semidefinite relaxation (SDR). 
	Then, we propose an efficient algorithm based on the above duality result for solving the considered problem. 
	The proposed algorithm judiciously exploits the special structure of an enhanced Karush-Kuhn-Tucker (KKT) conditions of the considered problem and approaches the solution that satisfies the enhanced KKT conditions via two fixed point iterations. 
	Two key features of the proposed algorithm are: 
	(1) it is able to detect whether the considered problem is feasible or not and find its globally optimal solution when it is feasible; 
	(2) it is highly efficient because both of the fixed point iterations in the proposed algorithm are linearly convergent and function evaluations in the fixed point iterations are computationally cheap. 
	Numerical results show the global optimality and efficiency of the proposed algorithm.   
\end{abstract}


\IEEEpeerreviewmaketitle

\section{Introduction}
\IEEEPARstart{L}{agrangian} duality \cite{boyd2004ConvexOptimization}, a principle that (convex) optimization problems can be viewed from either the primal or the dual perspective, is a powerful and vital tool in revealing the intrinsic structures of the optimization problems arising from engineering and further better solving the problems \cite{survey}. 
In practical engineering design, one is often interested in not only the numerical solution to the corresponding problems but also the specific structure of their optimal solutions. 
When a problem is formulated as a convex optimization problem, exploring its Lagrangian dual often reveals such structure. 
Knowing these solution structures in turn often leads to a better algorithm for solving the corresponding problem. 

The celebrated uplink-downlink duality \cite{
	rashid-farrokhi1998TransmitBeamformingPower, 
	boche2002GeneralDualityTheory, 
	schubert2004SolutionMultiuserDownlink,
	wiesel2006LinearPrecodingConic,
	viswanath2003SumCapacityVector
	}
	in the power control and beamforming design literature 
can be comprehensively understood and interpreted by Lagrangian duality \cite{song2007NetworkDualityMultiuser, yu2006UplinkdownlinkDualityMinimax}. 
The uplink-downlink duality refers to the fact that the minimum sum power required to achieve a set of signal-to-interference-plus-noise ratio (SINR) targets in the downlink channel is equal to that to achieve the same set of SINR targets in a virtual dual uplink channel, when the uplink and downlink channels are the conjugate transpose of each other. 
Usually, the virtual uplink beamforming problems, e.g., the sum transmission power minimization problems subject to users' SINR constraints, can be derived from some equivalent transformation of the Lagrangian dual of the downlink problem \cite{song2007NetworkDualityMultiuser} and solved globally and efficiently via the fixed point iteration algorithms \cite{rashid-farrokhi1998TransmitBeamformingPower,wiesel2006LinearPrecodingConic,visotsky1999OptimumBeamformingUsing,rashid-farrokhi1998JointOptimalPower}. 
The uplink-downlink duality thus enables efficient algorithms for solving the downlink problem via solving the relatively easy uplink (essentially dual) counterpart. 
This line of algorithms enjoys two key features: one is its high computational efficiency as the algorithm often involves cheap fixed point iterations only, and the other is its global optimality. 
Indeed, the Lagrangian duality and in particular the uplink-downlink duality based algorithms have been widely studied for solving power control and beamforming design problems in various communication networks; see \cite{
rashid-farrokhi1998TransmitBeamformingPower, 
boche2002GeneralDualityTheory, 
schubert2004SolutionMultiuserDownlink,
wiesel2006LinearPrecodingConic,
viswanath2003SumCapacityVector, 
song2007NetworkDualityMultiuser, 
vishwanath2003DualityAchievableRates,
yu2006UplinkdownlinkDualityMinimax, 
visotsky1999OptimumBeamformingUsing, 
rashid-farrokhi1998JointOptimalPower, 
bengtsson2002OptimumSuboptimumTransmit, 
schubert2005IterativeMultiuserUplink, 
cai2011MaxminWeightedSINR, 
yu2007TransmitterOptimizationMultiantenna, 
hammarwall2006DownlinkBeamformingIndefinite, 
dahrouj2010CoordinatedBeamformingMulticell}
and the references therein. 

Different from the above works where the degree of the cooperation between the base stations (BSs) is limited, this paper considers the cooperative cellular network where the users' data information are shared among the BSs via the fronthaul links and the joint processing is performed at the central processor (CP), which can effectively mitigate the inter-cell interference. 
Such network includes coordinated multipoint \cite{irmer2011CoordinatedMultipointConcepts}, distributed antenna system \cite{kerpez1996RadioAccessSystem}, cloud radio access network (C-RAN) \cite{simeone2016CloudRadioAccess,li2023AsynchronousActivityDetection,wang2024CovariancebasedActivityDetection}, and cell-free massive multi-input multi-output \cite{ngo2017CellfreeMassiveMIMO,ammar2022UsercentricCellfreeMassive,bjornson2020ScalableCellfreeMassivea,elhoushy2022CellfreeMassiveMIMO,zhao2023CommunicationefficientDecentralizedLinear,wang2024CovariancebasedActivityDetection,demir2021FoundationsUsercentricCellfree} as special cases. 
Despite the attractive advantages of full cooperation between the BSs, the cooperative cellular network puts heavy burden on the required fronthaul links. 
To tackle the above issue, transmission strategies of the BSs should be jointly designed along with the utilization of the fronthaul links \cite{peng2015FronthaulconstrainedCloudRadio}. 
Along this direction, a variety of solutions have been proposed under different design objectives and system settings; see \cite{
dai2014SparseBeamformingUsercentric, 
shi2014GroupSparseBeamforming,
park2013JointPrecodingMultivariate, 
park2014InterclusterDesignPrecoding,
jeon2016JointDesignsFronthaul, 
tang2019UserSelectionPower,
patil2014HybridCompressionMessagesharing, 
kang2016FronthaulCompressionPrecoding,
he2019HybridPrecoderDesign, 
kim2019JointDesignFronthauling,
ahn2020FronthaulCompressionPrecoding,
zhang2022BeamformingFronthaulCompression,
huang2022JointOptimizationWireless, 
kim2022CellfreeMmWaveMassive,
maryopi2021SumrateMaximizationUplink,
zhou2016FronthaulCompressionTransmit
} and the references therein. 

However, very few of the above works have exploited the Lagrangian duality or uplink-downlink duality in the cooperative cellular network
(possibly due to the reason that the optimization problem therein seems to be nonconvex). 
Note that the Lagrangian duality approach is a powerful method in deriving the uplink-downlink duality and developing efficient duality-based algorithms for solving various beamforming design problems when the fronthaul is assumed to have an infinite capacity \cite{visotsky1999OptimumBeamformingUsing, yu2006UplinkdownlinkDualityMinimax, wiesel2006LinearPrecodingConic, yu2007TransmitterOptimizationMultiantenna, song2007NetworkDualityMultiuser}. 
Then, an important question arises: when the fronthaul capacity is limited, does the Lagrangian duality approach still work in developing the efficient duality-based algorithm? 
Notice that the limited fronthaul capacity case differs significantly from the infinite fronthaul capacity case. 
Specifically, when formulating the beamforming design problem in the limited capacity case, additional fronthaul rate constraints come into play.  
These additional constraints have changed the problem's structure and introduced two technical challenges when applying the Lagrangian duality approach. 
Firstly, it is unclear whether the beamforming design problem with fronthaul rate constraints exhibits a zero duality gap. 
Secondly, it is also uncertain whether the problem with fronthaul rate constraints still has a favorable solution structure.
The goal of this paper is to answer the above questions, i.e., exploit the Lagrangian duality in the joint beamforming and compression design problem in the cooperative cellular network~(where the fronthaul capacity is limited) to reveal its special solution structure and further utilize them to develop efficient duality-based algorithms. 

\subsection{Prior Works}

Duality-based algorithms for the downlink beamforming problem in the conventional cellular network have been studied extensively in \cite{
	schubert2004SolutionMultiuserDownlink, 
	wiesel2006LinearPrecodingConic, 
	cai2011MaxminWeightedSINR, 
	yu2007TransmitterOptimizationMultiantenna, 
	hammarwall2006DownlinkBeamformingIndefinite
}. 
Assuming single-antenna users, a multi-antenna BS, and linear encoding and decoding strategies employed at the BS, the works in
\cite{
	rashid-farrokhi1998TransmitBeamformingPower, 
	boche2002GeneralDualityTheory, 
	schubert2004SolutionMultiuserDownlink,
	wiesel2006LinearPrecodingConic,
	viswanath2003SumCapacityVector, 
	song2007NetworkDualityMultiuser
	} 
showed that any downlink achievable SINR tuple can be achieved in the uplink under the same sum power constraint, and vice versa. 
Such uplink-downlink duality results enable efficient algorithms for solving the downlink beamforming problem. 
More specifically, the work \cite{schubert2004SolutionMultiuserDownlink} proposed an alternating optimization algorithm for solving the downlink beamforming problem and showed the global optimality of the algorithm. 
Instead of exactly solving the power allocation subproblem as in \cite{schubert2004SolutionMultiuserDownlink}, \cite{wiesel2006LinearPrecodingConic} proposed efficient fixed point algorithms for solving the downlink beamforming problem. 
Using the nonlinear Perron-Frobenius theory \cite{krause1986PerronStabilityTheorem}, the work \cite{cai2011MaxminWeightedSINR} proved that the fixed point iteration algorithm proposed in \cite{wiesel2006LinearPrecodingConic} is guaranteed to find the global solution. 
In addition, there have been works that exploit the uplink-downlink duality for downlink beamforming problems under various practical constraints, such as per-antenna power constraints 
\cite{yu2007TransmitterOptimizationMultiantenna} and indefinite shaping constraints \cite{hammarwall2006DownlinkBeamformingIndefinite}. 

In the cooperative cellular networks, the joint beamforming and compression problem, i.e., the joint design of the wireless transmission and the compression-based utilization of the fronthaul links, has been widely studied under various designing criteria and system settings \cite{
	park2013JointPrecodingMultivariate, 
	park2014InterclusterDesignPrecoding, 
	jeon2016JointDesignsFronthaul, 
	tang2019UserSelectionPower,
	kang2016FronthaulCompressionPrecoding, 
	kim2019JointDesignFronthauling, 
	ahn2020FronthaulCompressionPrecoding, 
	zhang2022BeamformingFronthaulCompression, 
	huang2022JointOptimizationWireless, 
	maryopi2021SumrateMaximizationUplink, 
	zhou2016FronthaulCompressionTransmit}. 
To fully utilize fronthaul links of finite capacities, an information-theoretically optimal compression strategy called multivariate compression was proposed in \cite{park2013JointPrecodingMultivariate}. 
Refs. \cite{park2013JointPrecodingMultivariate} and \cite{park2014InterclusterDesignPrecoding} studied the joint design of the beamformer and the covariance of the quantization noise under the assumption that the CP adopts the linear encoding strategy and the multivariate compression strategy to compress the signals before transmitting them to the relay-like BSs. 
More specifically, Ref. \cite{park2013JointPrecodingMultivariate} considered the weighted sum rate maximization problem with the total transmission power constraint at the CP and the fronthaul rate constraints of all relay-like BSs and proposed a successive convex approximation (SCA) algorithm for solving the considered problem. 
Ref. \cite{park2014InterclusterDesignPrecoding} further extended the above joint beamforming and compression design problem to the multi-cluster C-RAN case (with multiple CPs). 

In contrast to beamforming problems in conventional cellular networks, there is a scarcity of duality results and duality-based algorithms for the design of joint beamforming and compression in cooperative cellular networks. 
Indeed, most of the existing works use nonconvex optimization techniques (e.g., SCA) to tackle the joint beamforming and compression design problems in the cooperative cellular network. 
Recently, Ref. \cite{liu2021UplinkdownlinkDualityMultipleaccess}  generalized the uplink-downlink duality result from the conventional cellular network to the cooperative cellular network. 
Additionally, \cite{liu2021UplinkdownlinkDualityMultipleaccess} formulated a QoS-based joint beamforming and compression design problem and proposed an algorithm for solving it based on the established duality result. 
The algorithm in \cite{liu2021UplinkdownlinkDualityMultipleaccess} first obtains the optimal downlink beamformers by solving the uplink problem via fixed point iterations and then solves the downlink joint power control and compression problem with fixed beamformers. 

\subsection{Our Contributions}
In this paper, we consider the same QoS-based joint beamforming and compression design problem (see problem \eqref{equ:original obp} further ahead) as in \cite{liu2021UplinkdownlinkDualityMultipleaccess} but make further progress in developing the duality result and designing the duality-based algorithm. 
The main contributions of this paper are twofold.  

\begin{itemize}
	\item \emph{New Lagrangian Duality Result.} We establish the tightness of the semidefinite relaxation (SDR) of the considered problem and thus the equivalence of the two problems. 
	This result further implies that the dual problem of the considered problem and its SDR are the same.  
	This Lagrangian duality result significantly facilitates the algorithmic design and plays a central role in the proposed algorithm for solving the problem. 
	Our duality result is sharply different from the established duality result in \cite{liu2021UplinkdownlinkDualityMultipleaccess}, where the problem \eqref{equ:original obp} with fixed beamformers is considered. 
	The problem \eqref{equ:original obp} with fixed beamformers is already in a convex form (after some algebraic manipulation), while the problem \eqref{equ:original obp} itself is not, and whether it admits a convex reformulation is an open question in \cite[Section IX-B]{liu2021UplinkdownlinkDualityMultipleaccess}. 
	This makes our Lagrangian duality result nontrivial.

	\item \emph{Efficient Fixed Point Iteration Algorithm.} Based on the established duality result, we propose an efficient algorithm for solving the SDR of the considered problem. 
	The basic idea of the proposed algorithm is to solve the enhanced Karush-Kuhn-Tucker (KKT) conditions, which incorporate the special structures of the problem into the classical KKT conditions. 
	In particular, the proposed algorithm first solves the enhanced KKT conditions involving the dual variables via a fixed point iteration 
	and then solves the enhanced KKT conditions involving the primal variables via another fixed point iteration. 
	Two key features of the proposed algorithm are as follows: (1) it is guaranteed to find the global solution of the problem when it is feasible and is able to detect the infeasibility of the problem when it is not; (2) it is highly efficient because both fixed point iterations in the proposed algorithm enjoy linear convergence rates, and each update of the variables in fixed point iterations is computationally cheap. 
	The proposed algorithm leverages more Lagrangian duality relationship as compared with that in \cite{liu2021UplinkdownlinkDualityMultipleaccess}. 
	In particular, after obtaining the dual variables, our algorithm recovers the primal variables (e.g., power control vector and compression covariance matrix) via the fixed point iteration. 
	This is different from the algorithm in \cite{liu2021UplinkdownlinkDualityMultipleaccess}, which requires solving the downlink joint power control and compression problem with fixed beamformers from scratch.     
	This key difference is due to the new Lagrangian duality result and it makes our proposed algorithm significantly outperform the algorithm in \cite{liu2021UplinkdownlinkDualityMultipleaccess} in terms of the computational efficiency.
\end{itemize}

In our prior work \cite{fan2022EfficientlyGloballySolving}, we presented an efficient fixed point iteration algorithm for solving the QoS-based joint beamforming and compression design problem. 
The present paper, however, is a significant extension of \cite{fan2022EfficientlyGloballySolving}. 
First, we provide crucial details on the convergence proof of the fixed point iterations, which were missing in our prior work. 
Second, we show the linear convergence rate of the proposed fixed point iteration algorithm and study the behaviors of the proposed algorithm when the considered problem is infeasible. 
These theoretical results are completely new compared with our prior work. 
Third, we conduct detailed numerical experiments that compare the proposed algorithm with the state-of-the-art (SOTA) benchmarks. 

\subsection{Organization}
We adopt the following notations in this paper. 
We use $ \mathcal{S}_{++}^n $ to denote the set of all $ n\times n $ positive definite matrices, $ \mathcal{S}_{+}^n $ to denote the set of all $ n\times n $ positive semidefinite matrices, $ \mathbb{R}_{++}^n $ to denote the $n$-dimensional positive orthant, and $ \mathbb{R}_{+}^n $ to denote the $n$-dimensional nonnegative orthant. 
The order relationship between two vectors shall be understood component-wise. 
For any matrix $\m A$, $\m A^\hermitian$, $\m A^\transpose$, and $\m A^{-1}$ denote the conjugate transpose, transpose, and pseudo-inverse of $\m A$, respectively; 
$\m A^{(m, n)}$ denotes the entry on the $m$-th row and the $n$-th column of $\m A;$ and
$\m A^{(m_1:m_2, n_1:n_2)}$ denotes a submatrix of $\m A$ defined by 
\begin{equation*}
	\begin{bmatrix}
		\m A^{(m_1,n_1)} &\m A^{(m_1,n_1+1)} & \cdots &\m A^{(m_1,n_2)} \\
		\m A^{(m_1+1,n_1)} &\m A^{(m_1+1,n_1+1)} & \cdots &\m A^{(m_1+1,n_2)} \\
		\vdots&\vdots&\ddots&\vdots \\
		\m A^{(m_2,n_1)} &\m A^{(m_2,n_1+1)} &\cdots &\m A^{(m_2,n_2)}
	\end{bmatrix}. 
\end{equation*}
For two $n \times n$ matrices $\m A_1$ and $\m A_2$, $ \m A_1 \succeq \m A_2 $ and $ \m A_1 \succ \m A_2 $ denote that $ \m A_1 - \m A_2 \in \mathcal{S}_{+}^n $ and $ \m A_1 - \m A_2 \in \mathcal{S}_{++}^n $, respectively. 
We use $\cCN(\m 0 , \bar{\m Q})$ to denote the $n$-dimensional complex Gaussian distribution with zero mean and covariance $\bar{\m Q} \in \mathcal{S}_{+}^n$. Finally, we use $\m{I}$ to denote the identity matrix of an appropriate size, $\m 0 $ to denote an all-zero matrix of an appropriate size, $ \v{e}_m $ to denote the $m$-th column vector of $\m I$, and $\m{E}_m$ to denote $ \v{e}_m \v{e}_m^\hermitian $.

\section{System Model and Problem Formulation}
\subsection{System Model}

Consider a cooperative cellular network consisting of one CP and $M$ single-antenna relay-like BSs (which will be called relays for short later), 
which cooperatively serve $K$ single-antenna users. 
In such network, the users and the relays are connected by noisy wireless channels, and the relays and the CP are connected by noiseless fronthaul links of finite capacities. Let $\cM = \left\{ 1, 2, \d, M \right\}$ and $\cK = \left\{ 1, 2, \d, K \right\}$ denote the sets of the relays and the users, respectively. 

We first introduce the compression model from the CP to the relays. The beamformed signal at the CP is $\tilde{\v{x}} = \sum_{k\in\cK} \v v_k s_k$, 
where $\v v_k = [v_{k,1}, v_{k,2}, \d , v_{k,M}]^\transpose$ is the $M \times 1$ beamforming vector and $s_k\sim \cCN(0, 1)$ is the information signal for user $k$. 
Because of the limited capacities of the fronthaul links, 
the signal from the CP to the relays need to be first compressed before transmitted. 
Using compression with the Gaussian test channel \cite{park2013JointPrecodingMultivariate, park2014FronthaulCompressionCloud},  the compression error is modeled as a Gaussian variable, independent of $\{s_k\}$, i.e., $\v{e} = [e_1, e_2, \d, e_M]^\transpose \sim \cCN(\m 0 , \m Q),$ where $e_m$ denotes the error for compressing signals to relay $m$, and $\m Q$ is the covariance matrix of the compression noise. 
The transmitted signal of relay $m$ is 
\begin{equation}
	x_m = \sum_{k\in \cK} v_{k,m}s_k + e_m, \fm. 
	\label{equ:model_R}
\end{equation} 
Then the received signal of user $ k $ is 
\begin{equation*}
	y_k = \sum_{m\in \cM} h_{k,m} x_m + z_k, \fk, 
\end{equation*}
where $h_{k,m}$ is the channel coefficient from relay $m$ to user $k$, and $\left\{z_1, z_2, \d, z_K\right\}$ are independent and identically distributed (i.i.d.) additive complex Gaussian noise distributed as $\cCN(0, \sigma_k^2).$

Under the above model, the received signal at user $k$ is 
\begin{equation*}
	y_k = \v{h}_k^\hermitian \left( \sum_{i\in\cK} \v v_i s_i \right) + \v{h}_k^\hermitian \v{e} + z_k, \fk, 
	\label{equ:DBM}
\end{equation*}
where $\v{h}_k = [h_{k,1}, h_{k,2}, \d, h_{k,M}]^\hermitian$ is the channel vector of user $k$. Then, the total transmit power of all the relays is 
	\begin{equation*}
		\sum_{k\in\cK} \|\v v_k\|^2 + \tr(\m Q). 
		\label{equ:model_power}
	\end{equation*}
The SINR of user $k$ is
	\begin{equation*}
		\gamma_k(\{\v v_k\}, \m Q) = \frac{|\v{h}_k^\hermitian \v v_k|^2}{\sum_{j\neq k} |\v{h}_k^\hermitian \v v_j|^2 + \v{h}_k^\hermitian \m Q \v{h}_k + \sigma_k^2},~\forall~k\in\cK. 
	\end{equation*}
In order to fully utilize fronthaul links of finite capacities, we adopt the information-theoretically optimal multivariate compression strategy \cite{park2013JointPrecodingMultivariate} to compress the signals from the CP to the relays. 
Without loss of generality, we assume that the compression order is from relay $ M $ to relay $ 1 $. Then the compression rate of relay $m$ is given by 
\begin{equation}
\label{equ:def_Cm}
	\begin{aligned}
		&C_m(\{\v v_k\}, \m Q) \\
		={}& \log_2 \left( \frac{\sum_{k\in\cK} |v_{k,m}|^2 + \m Q^{(m,m)}}{q_m} \right), \fm,
	\end{aligned}
\end{equation}
where $$q_m = \m Q^{(m,m)} - \m Q^{(m,m+1:M)}(\m Q^{(m+1:M,m+1:M)})^{-1} \m Q^{(m+1:M,m)}$$ is the (generalized) Schur complement of $\m Q^{(m+1:M,m+1:M)}$ in $\m Q^{(m:M,m:M)}$. 
Using the information theoretic results \cite{elgamal2011NetworkInformationTheory, park2013JointPrecodingMultivariate}, for any given covariance matrix $\m Q$ and the beamforming vectors $\{\v v_k\}$, we can find a compression strategy such that $\{C_m(\{\v v_k\}, \m Q)\}$ bits can be transmitted over the fronthaul link without error per second per Hz.

Notice that when $\m Q$ is singular, i.e., $q_m = 0$ for some $m$, $C_m$ in \eqref{equ:def_Cm} is not well defined. 
For the completeness of the definition and the closeness of the feasible region of the considered problem \eqref{equ:original obp} further ahead, we use the following definition when $q_m = 0$ for some $m\in\cM$: 
set $C_m = 0$ if ${\sum_{k\in\cK} |v_{k,m}|^2 + \m Q^{(m,m)}} = 0$, i.e., $v_{k,m} = 0$ for all $k\in\cK$ and $\v Q^{(1:M,m)} = (\v Q^{(m,1:M)})^\hermitian = \v 0$; set $C_m = +\infty$ otherwise. 
When the singular case $C_m = 0$ happens, relay $m$ does not play any role in the whole transmission process. 

\subsection{Problem Formulation}
Given a set of SINR targets for the users $\left\{ \bar\gamma_k \right\}$ and a set of fronthaul capacities for the relays $\left\{ \bar C_m \right\}$, 
we aim to minimize the total transmit power of all the relays, as shown in \eqref{equ:original obp}:
\begin{equation}
	\begin{aligned}
		\min_{\{\v v_k\}, \m Q\succeq \m 0 } &\quad \sum_{k\in\cK} \|\v v_k\|^2 + \tr(\m Q)\\
		\st~~~~ &\quad \gamma_k(\{\v v_k\}, \m Q) \geq \bar \gamma_k, \fk,\\
		&\quad C_m(\{\v v_k\}, \m Q) \leq \bar{C}_m, \fm{}. 
	\end{aligned}
	\label{equ:original obp}
\end{equation}

\begin{figure*}[!t]
	\begin{equation}
		\begin{aligned}
			\min_{\{\v v_k\}, \m Q\succeq \m 0 }&\quad \sum_{k\in\cK} \|\v v_k\|^2 + \tr(\m Q)\\
			\st~~~~ &\quad  \frac{1}{\bar \gamma_k}|\v h_k^\hermitian \v v_k|^2 - \left(\sum_{j\neq k} |\v h_k^\hermitian \v v_j|^2 + \v h_k^\hermitian \m Q \v h_k + \sigma_k^2\right) \geq 0, \fk, \\
			&\quad 2^{\bar C_m} \begin{bmatrix}
				\m 0 & \m 0  \\
				\m 0 & \m Q^{(m:M, m:M)}
			\end{bmatrix} - \left(\sum_{k\in\cK} |v_{k,m}|^2+ \m Q^{(m,m)}\right) \m{E}_m \succeq \m 0 , \fm{}. 
		\end{aligned} \tag{P}
	\end{equation}
	\hrulefill
\end{figure*}
By some algebraic manipulations, we can show that problem \eqref{equ:original obp} is equivalent to problem (P) at the top of the next page. 
Please refer to Appendix~\ref{apd:equivalence} in the Supplementary Material for the details of the proof. 
Notice that problem (P) with fixed beamformers, which is a semidefinite program (SDP), is studied in \cite{liu2021UplinkdownlinkDualityMultipleaccess} to derive their uplink downlink duality results. 
Under the assumption of strict feasibility, problem (P) with fixed beamformers, being an SDP, enjoys strong duality.
However, problem (P) itself is not in a convex form, and the existence of its convex reformulation is an unanswered question \cite{liu2021UplinkdownlinkDualityMultipleaccess}. 
Notice that the convex reformulation technique proposed in \cite{wiesel2006LinearPrecodingConic}, which turns the SINR constraints into a set of second order cone constraints via performing the square root operation on both sides of the SINR constraints, cannot be applied to problem (P) because the additional variable $\m Q$ makes the resulting constraints nonconvex. 

In the following section, we will give a convex reformulation of problem (P) by deriving the SDR of problem (P) and showing its tightness. 
Then, we will design an efficient algorithm for globally solving problem (P) by solving the enhanced KKT conditions of its SDR. 

\section{SDR of (P) and Its Tightness}
Problem (P) is a quadratically constrained quadratic program of $\{\v v_k\}$ and an SDP of $\m Q$. 
A well-known technique to tackle such problem is the SDR \cite{luo2010SemidefiniteRelaxationQuadratic}. 
In particular, setting $\m V_k = \v v_k \v v_k^\hermitian$ and relaxing the rank-one constraint of $\m V_k$ for all $k\in\cK$ in problem (P), 
we obtain the SDR of problem (P): 
\begin{equation}\label{SDR} 
	\begin{aligned}
		\min_{\{\m V_k\}, \m Q \succeq \m 0 } &\quad \sum_{k\in\cK} \tr(\m V_k) + \tr(\m Q)\\
		\st~~~~ &\quad a_k(\left\{\m V_k\right\},\m Q) \geq 0,\fk, \\
		&\quad \m B_m(\left\{\m V_k\right\},\m Q) \succeq \m 0 ,\fm{},\\
		&\quad \m V_k \succeq \m 0  ,\fk, 
	\end{aligned}
\end{equation}
where 
\begin{equation*}
	\begin{aligned}
	a_k(\left\{\m V_k\right\},\m Q) & = \frac{1}{\bar \gamma_k} \v h_k^\hermitian \m V_k \v h_k - \sum_{j\neq k} \v h_k^\hermitian \m V_j \v h_k- \v h_k^\hermitian \m Q \v h_k - \sigma_k^2\\ 
		\m B_m(\left\{\m V_k\right\},\m Q) & = 2^{\bar C_m} \begin{bmatrix}
			\m 0 & \m 0  \\
			\m 0 & \m Q^{(m:M, m:M)}
		\end{bmatrix} \\
		& \pushright{-~\left(\sum_{k\in\cK} \m V_k^{(m,m)}+ \m Q^{(m,m)}\right) \m{E}_m. } \\ 
	\end{aligned}
\end{equation*} 
Since the problem~\eqref{SDR} is convex, we consider its Lagrangian dual problem, which is given by
	\begin{equation}\label{dual}
		{
			\begin{aligned}
				\max_{\v \beta \geq \v 0, \Lm} &\quad \sum_{k\in\cK} \beta_k \sigma_k^2 \\
				\st~~~~ &\quad \m C_k(\v \beta, \left\{\m{\Lambda}_m\right\}) - \frac{1}{\bar \gamma_k}\beta_k \v h_k \v h_k^\hermitian \succeq \m 0 ,\fk,\\
				&\quad	 \m{D}(\v \beta,\left\{\m{\Lambda}_m\right\}) \succeq \m 0 ,\\
				&\quad \m{\Lambda}_m \succeq \m 0  ,\fm{},
			\end{aligned}
		}
	\end{equation}
where $ \v \beta = [\beta_1, \beta_2, \d, \beta_K]^\transpose$ with $\beta_k$ being the dual variable associated with the $k$-th SINR constraint in problem~\eqref{SDR}, $\m{\Lambda}_m$ is the dual variable associated with the $m$-th fronthaul rate constraint in problem~\eqref{SDR}, and
\begin{equation*}
	\begin{aligned}
		\m C_k(\v \beta,\left\{\m{\Lambda}_m\right\}) &= \m I + \sum_{j\neq k} \beta_j \v h_j \v h_j^\hermitian + \sum_{m\in\cM} \m \Lambda_m^{(m,m)} \m E_m, \\
		\m{D}(\v \beta,\left\{\m{\Lambda}_m\right\}) &= \m{I} + \sum_{k\in\cK} \beta_k \v h_k \v h_k^\hermitian + \sum_{m\in\cM} \m \Lambda_m^{(m,m)} \m E_m \\
		& ~~~~~~~~~~~~~ - \sum_{m\in\cM} 2^{\bar C_m} \left[ \begin{matrix}
			\m 0 & \m 0  \\
			\m 0 & \m{\Lambda}_m^{(m:M,m:M)}
			\end{matrix} \right].
	\end{aligned}
\end{equation*}

The tightness of the SDR, i.e., whether the SDR problem admits a rank-one solution, is an important line of research on the SDR \cite{luo2010SemidefiniteRelaxationQuadratic,lu2019TightnessNewEnhanced,liu2022CramerRaoBoundOptimization}. 
Here, we prove the tightness of the SDR in \eqref{SDR} under the strict feasibility assumption. 
This reveals the hidden convexity in the seemingly nonconvex problem (P) and shows that the problem admits a convex reformulation, which answers a question in \cite[Section IX-B]{liu2021UplinkdownlinkDualityMultipleaccess}. 

\begin{theorem}
	\label{thm:tight}
	Suppose that problem~\eqref{SDR} is strictly feasible, and let $(\Vk, \m Q)$ be its solution. Then $\m V_k$ is of rank one for all $k\in\cK$. 
\end{theorem}
\begin{proof}
    Because of the optimality of $(\Vk, \m Q)$ and strict feasibility of  problem~\eqref{SDR}, there must exist dual multipliers $\v \beta$ and $\Lm$ such that the Karush-Kuhn-Tucker (KKT) conditions of problem~\eqref{SDR} hold. 
	In particular, the complementary slackness conditions $$\tr \left( \m V_k \left(\m C_k(\v \beta,\left\{\m{\Lambda}_m\right\}) - \frac{1}{\bar \gamma_k}\beta_k \v h_k \v h_k^\hermitian\right) \right)= 0, \fk$$ hold. 
	Since $\m C_k(\v \beta,\left\{\m{\Lambda}_m\right\})$ is positive definite and $\v h_k \v h_k^\hermitian$ is rank-one and positive semidefinite, it follows that 
	$$\rk\left(\m C_k(\v \beta,\left\{\m{\Lambda}_m\right\}) - \frac{1}{\bar \gamma_k} \beta_k \v h_k \v h_k^\hermitian \right) \geq M-1, \fk, $$ 
	which, together with the complementary slackness condition and the rank inequality, implies that 
	$\rk(\m V_k)\leq 1 \text{~for~all~} k \in \cK$. 
	If $\m V_k = 0$, i.e., user $ k $ is assigned with a zero beamformer, then its corresponding SINR constraint will be violated. 
	Therefore, all optimal $\m V_k$ are rank-one. 
	The proof is complete. 
\end{proof}

Two remarks on Theorem~\ref{thm:tight} are as follows. 
First, Theorem~\ref{thm:tight} offers a way of globally solving problem (P) via solving its SDR \eqref{SDR}, which provides an important benchmark for performance evaluation of other algorithms for solving problem (P). 
Second, assuming that problem~\eqref{SDR} is strictly feasible, it is well known that the KKT conditions are sufficient and necessary for the global solution of problem~\eqref{SDR} \cite[Section 5.2]{boyd2004ConvexOptimization}. 
It will become more clear that the KKT conditions play a central role in solving problem~\eqref{SDR}. 
In the following, we shall design an efficient fixed point algorithm for solving the KKT conditions of problem~\eqref{SDR}. 

\section{Proposed Fixed Point Iteration Algorithm}
In this section, we first combine the special structures of the solution of the problem~\eqref{SDR} with its KKT conditions to derive a new set of conditions, which is referred to as the enhanced KKT conditions. 
Then, we present an efficient way of solving the enhanced KKT conditions via two fixed point iterations and thus an efficient algorithm for solving the problem~\eqref{SDR}. 

\subsection{Enhanced KKT Conditions}
First, recall that the KKT conditions of an SDP consist of the dual feasibility conditions, the primal feasibility conditions, and the complementary slackness conditions. 
Specifically, supposing $\Vk$ and $\m Q$ are the primal solutions, and $\v \beta$ and $\Lm$ are the dual solutions, the complementary conditions of problem~\eqref{SDR} are given by 
\begin{equation}
	\tr (\m Q \m D(\v \beta, \Lm)) = \m 0, 
	\label{equ:QD}
\end{equation}
\begin{equation}
	\begin{aligned}
		\label{equ:slack_1}
		\tr \left( \m V_k \left( \m C_k(\v \beta,\left\{\m{\Lambda}_m\right\}) - \frac{1}{\bar \gamma_k} \beta_k \v h_k \v h_k^\hermitian\right) \right) =0,\fk, 
	\end{aligned}
\end{equation}
and
\begin{equation}
	\label{equ:slack_2}
	\tr \left( \m{\Lambda}_m \m B_m(\left\{\m V_k\right\},\m Q) \right) = 0, \fm{}. 
\end{equation}
Note from problem (P) that if all the relays are utilized (i.e., $\sum_{k\in\cK} |v_{k,m}|^2 > 0$ for all $m \in \cM$) at the optimal solution, then $\m Q$ must be positive definite (otherwise the fronthaul rate constraint will be violated). 
Combining this with Eq. \eqref{equ:QD} gives $\m D(\v \beta, \Lm) = \m 0$. 
Using this and the structure of $\m C_k(\v \beta,\left\{\m{\Lambda}_m\right\}) -  \beta_k/\bar \gamma_k \v h_k \v h_k^\hermitian$, the dual feasibility conditions are enhanced into the following conditions: 
\begin{numcases}{}
	\label{equ:dual_const2}
	\m D(\v \beta, \Lm) = \m 0, \\
	\label{equ:dual_var2}
	\left.
	\begin{aligned}
		&\rk(\m{\Lambda}_m)\leq 1,~\m{\Lambda}_m \succeq \m 0 ,\fm{}, \\
		&\m{\Lambda}_m^{(1:m-1,1:m)}= \m 0 ,~\m{\Lambda}_m^{(m:M,1:m-1)} = \m 0 ,\fm{},
	\end{aligned}
	\right\}\\
	\label{equ:dual_var1}
	\v \beta\geq \v 0, \\
	\label{equ:dual_const1}
	\left.
		\begin{aligned}
			&\rk\left(\m C_k(\v \beta, \left\{\m{\Lambda}_m\right\}) - \beta_k/\bar \gamma_k\v h_k \v h_k^\hermitian\right)\\
			& \pushright{=M-1, \fm,} \\
			&\m C_k(\v \beta,\left\{\m{\Lambda}_m\right\}) -  \beta_k/\bar \gamma_k \v h_k \v h_k^\hermitian \succeq \m 0 , \fm{}, 
		\end{aligned}
	\right\}
\end{numcases}
where \eqref{equ:dual_var2} is shown in Appendix~\ref{apd:kkt}. 
Furthermore, the primal feasibility conditions are enhanced into the following conditions: 
\begin{numcases}{}
	\label{equ:pri_var1}
	\m V_k\succeq \m 0 ,~\rk(\m V_k) = 1, \fk, \\
	\label{equ:pri_const1}
	a_k(\left\{\m V_k\right\},\m Q)=0,\fk,\\
	\label{equ:pri_const2}
	\m B_m(\left\{\m V_k\right\},\m Q)\succeq \m 0 , \fm{}, \\
	\label{equ:pri_var2}
	\m Q \succeq \m 0, 
\end{numcases}
where \eqref{equ:pri_var1} is shown by Theorem \ref{thm:tight}, and \eqref{equ:pri_const1} is shown in Appendix~\ref{apd:kkt}. 

The enhanced KKT conditions, i.e., Eqs.~\eqref{equ:slack_1}--\eqref{equ:pri_var2}, define a nonempty subset of KKT points under the assumption that any optimal $\m Q$ is positive definite. 
In Appendix~\ref{apd:PD} in the Supplementary Material, we prove that any optimal $\m Q$ is positive definite with probability one under a mild assumption on the random channels, which provides a theoretical justification on our assumption of the positive definiteness of $\m Q$ during the algorithm development. 

\begin{figure}[!t]
	\centering
	\includegraphics[width=0.32\textwidth]{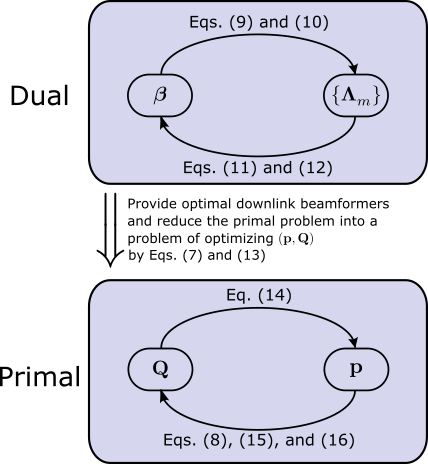}
	\caption{The flow chart of the proposed algorithm for solving the enhanced KKT conditions. }
	\label{fig:illustration_kkt}
\end{figure}
Next, we shall design an algorithm for solving the enhanced KKT conditions. 
The basic idea is to first solve the enhanced dual feasibility conditions, i.e., Eqs.~\eqref{equ:dual_const2}--\eqref{equ:dual_const1}, for the dual variables $\v \beta$ and $ \Lm $; and then plug them into the rest of the enhanced KKT conditions to solve for the primal variables $\Vk$ and $\m Q$. 
The flow chart of the proposed algorithm for solving the enhanced KKT conditions is illustrated in Fig.~\ref{fig:illustration_kkt}. 

\subsection{Solving for the Dual Variables} \label{ss:dual}

\subsubsection{Solving Eqs. \eqref{equ:dual_const2} and \eqref{equ:dual_var2} for $ \Lm $}
\label{sss:solvingforLm}
Suppose that $\v \beta$ is given, we first find $\Lm$ that satisfy Eqs.~\eqref{equ:dual_const2} and \eqref{equ:dual_var2}. 
Define 
\begin{equation}
    \label{equ:def_Gamma}
    \m \Gamma(\v \beta) = \m{I} + \sum_{k\in\cK} \beta_k \v h_k \v h_k^\hermitian. 
\end{equation}
Then Eq. \eqref{equ:dual_const2} is equivalent to
\begin{equation}
    \sum_{m\in\cM} 2^{\bar C_m} \left[ \begin{matrix}
			\m 0 & \m 0  \\
			\m 0 & \m{\Lambda}_m^{(m:M,m:M)}
			\end{matrix} \right]  - \sum_{m\in\cM} \m \Lambda_m^{(m,m)} \m E_m
			= \m{\Gamma}(\v \beta). 
	\label{equ:D_changed}
\end{equation}
We know from the special properties of $\Lm$ in Eq. \eqref{equ:dual_var2} that only $\m{\Lambda}_1$ affects the first row and column of matrix $\m{\Gamma}(\v \beta).$ 
Therefore, the entries in the first row of $\m{\Lambda}_1$ should be 
\[\left[ \begin{matrix}
	\frac{1}{2^{\bar C_1} - 1}\m{\Gamma}(\v \beta)^{(1,1)}, \frac{1}{2^{\bar C_1}}\m{\Gamma}(\v \beta)^{(1,2:M)}
\end{matrix} \right]. \]
Since $\m{\Lambda}_1$ is of rank one, we can further obtain all entries of $\m{\Lambda}_1$ based on its entries in the first row, which is 
\begin{equation}
	\m{\Lambda}_1 = 
	\begin{bmatrix}
		\frac{1}{2^{\bar C_1}-1} \m{\Gamma}(\v \beta)^{(1,1)}
		&\frac{1}{2^{\bar C_1}} \m{\Gamma}(\v \beta)^{(1,2:M)}  \\
		\frac{1}{2^{\bar C_1}} \m{\Gamma}(\v \beta)^{(2:M,1)}  
		&\frac{2^{\bar C_1}-1}{4^{\bar C_1}} \frac{\m{\Gamma}(\v \beta)^{(2:M,1)} \m{\Gamma}(\v \beta)^{(1,2:M)}}{\m{\Gamma}(\v \beta)^{(1,1)}}
	\end{bmatrix}. 
	\label{equ:Gamma2Lambda1}
\end{equation} 
After $\m{\Lambda}_1$ is obtained, we can subtract all terms related to $\m{\Lambda}_1$ from both sides of \eqref{equ:D_changed} and denote the known right-hand side as $ \m{\Gamma}_2(\v \beta) $. 
Then, after the subtraction, \eqref{equ:D_changed} becomes
\begin{equation}
	\sum_{m = 2}^M 2^{\bar C_m} \left[ \begin{matrix}
			\m 0 & \m 0  \\
			\m 0 & \m{\Lambda}_m^{(m:M,m:M)}
			\end{matrix} \right]  - \sum_{m = 2}^M \m \Lambda_m^{(m,m)} \m E_m 
			= \m{\Gamma}_2(\v \beta). 
	\label{equ:D_changed2}
\end{equation}

For general $ 2 \leq m \leq M $, denote the right-hand side of \eqref{equ:D_changed} after $ \m{\Lambda}_1, \m{\Lambda}_2, \d, \m\Lambda_{m-1} $ are solved and subtracted from both sides as $ \m\Gamma_m(\v \beta) $. 
By comparing the $m$-th row of the obtained equation and using the rank-one property of $\m\Lambda_{m}$, we can obtain $ \m\Lambda_{m}, $ whose nonzero part $\m\Lambda_m^{(m:M,m:M)}$ is given in \eqref{equ:Gamma2Lambdam} at the top of the next page. 
\begin{figure*}[!t]
    \begin{equation}
	\begin{bmatrix}
		\frac{1}{2^{\bar C_m}-1} \m{\Gamma}_m(\v \beta) ^{(m,m)}  
		&\frac{1}{2^{\bar C_m}} \m{\Gamma}_m(\v \beta) ^{(m,m+1:M)}  \\
		\frac{1}{2^{\bar C_m}} \m{\Gamma}_m(\v \beta) ^{(m+1:M,m)}  
		&\frac{2^{\bar C_m}-1}{4^{\bar C_m}} \frac{\m{\Gamma}_m(\v \beta) ^{(m+1:M,m)} \m{\Gamma}_m(\v \beta) ^{(m,m+1:M)}}{\m{\Gamma}_m(\v \beta) ^{(m,m)}}
	\end{bmatrix}. 
	\label{equ:Gamma2Lambdam}
\end{equation} 
\hrulefill
\end{figure*}
We can repeat the above procedure until all $\{\m \Lambda_m\}$ are obtained. 
These solutions, depending on the given $\v \beta$, are denoted as $\{\m \Lambda_m (\v \beta)\}$. 
Please refer to Appendix~\ref{apd:solve_Lambda} in the Supplementary Material for the details of the above procedure. 

\subsubsection{Solving Eqs.~\eqref{equ:dual_var1} and \eqref{equ:dual_const1} for $ \v \beta $}
Now suppose that $ \Lm $ are given, we would like to find $ \v \beta $ that satisfies Eqs. \eqref{equ:dual_var1} and \eqref{equ:dual_const1}. 
Since $\m C_k(\v \beta,\left\{\m{\Lambda}_m\right\}) \succ \m 0 $ and $\v h_k \v h_k^\hermitian\succeq \m 0 $ is rank-one, there exists a unique $\beta_k$ such that one and only one eigenvalue of $\m C_k(\v \beta,\left\{\m{\Lambda}_m\right\})-\beta_k / {\bar \gamma_k}\v h_k \v h_k^\hermitian$ is equal to zero.  
Such $\beta_k$ admits the following closed-form solution: 
\begin{equation}
	\beta_k\left(\Lm, \v \beta\right) = \frac{\bar \gamma_k}{\v{h}_k^\hermitian \m C_k(\v \beta, \left\{\m{\Lambda}_m\right\})^{-1} \v h_k}, \fk. 
	\label{equ:L2b}
\end{equation}

\subsubsection{Dual Fixed Point Iteration}
From the above discussion, we know that 
if $\v \beta$ is known, one can get $\left\{ \m \Lambda_m (\v \beta) \right\}$ such that Eqs. \eqref{equ:dual_const2} and \eqref{equ:dual_var2} hold. 
Plug this solution $\left\{ \m \Lambda_m (\v \beta) \right\}$ into \eqref{equ:L2b}. 
Then, if one can find $\v \beta$ that satisfy
\begin{equation}
	\beta_k = I_k\left(\v \beta\right) \triangleq \beta_k\left(\left\{ \m \Lambda_m (\v \beta) \right\}, \v \beta\right),\fk, 
	\label{equ:raw_dual_iter}
\end{equation}
all Eqs.~\eqstd{} are satisfied. 
If we define
$I(\v \beta) = \left[I_1(\v \beta),I_2(\v \beta),\d,I_K(\v \beta)\right]^\transpose,$ 
then solving \eqref{equ:raw_dual_iter} is to find the fixed point of the function $ I(\cdot) $, namely solving 
\begin{equation}
	\v \beta = I(\v \beta). 
	\label{equ:dual_fix_point}
\end{equation}

It is worth highlighting that the computational cost of evaluating the function $I(\cdot)$ in \eqref{equ:dual_fix_point} is quite cheap. 
The dominant computation is to compute $ \m C_k^{-1} \v{h}_k $ for $ k \in \cK $, where $\m C_k$ is an $M \times M$ positive definite matrix and the corresponding arithmetic complexity is $\mathcal{O}(KM^3)$. 
By using the matrix inversion formula, one can instead solve $K$ of linear systems with the same coefficient matrix of size $M \times M$ to evaluate $I(\cdot)$. 
Hence, the arithmetic complexity can be reduced to $\mathcal{O}(M^2\max\{K,M\})$. 
Furthermore, as will be shown later in Theorem~2, Eq. \eqref{equ:dual_fix_point} can be easily solved via the dual fixed point iteration 
\begin{equation}
	\v \beta^{(i+1)} = I(\v \beta^{(i)})
	\label{equ:dual_iter}
\end{equation}
with a linear convergence rate. 
This fact, together with the cheap evaluation of the function $I(\cdot)$ at each iteration, shows that the above fixed point iteration in \eqref{equ:dual_iter} provides an efficient way of solving the enhanced dual feasibility conditions, i.e, Eqs.~\eqstd{}. 

\subsection{Solving for the Primal Variables}\label{ss:primal}
Suppose that we already have $\v \beta$ and $\Lm$ that satisfy the enhanced dual feasibility conditions. 
We still need to find $\Vk$ and $\m Q$ that satisfy the rest of the enhanced KKT conditions. 
By Eq. \eqref{equ:pri_var1}, let $\m V_k = p_k \hat{\v v}_k \hat{\v v}_k^\hermitian$ with some $\|\hat{\v v}_k\|=1$. 
Then Eq. \eqref{equ:slack_1} becomes
$\hat{\v v}_k^\hermitian \left(\m C_k- \beta_k/\bar \gamma_k \v h_k \v h_k^\hermitian\right) \hat{\v v}_k = 0.$ 
Combining this with Eq. \eqref{equ:dual_const1} gives
$$\left(\m C_k- \beta_k/\bar \gamma_k \v h_k \v h_k^\hermitian\right) \hat{\v v}_k=\m C_k \hat{\v v}_k - \beta_k/\bar\gamma_k  \left(\v{h}_k^\hermitian \hat{\v v}_k\right) \v{h}_k = \m 0 .$$
Hence, $\hat{\v v}_k$ can be solved explicitly as follows: 
\begin{equation}
	\hat{\v v}_k = \frac{\m C_k^{-1} \v{h}_k}{\left\|\m C_k^{-1} \v{h}_k\right\|}, \fk.
\label{equ:pri_beamforming_direction}
\end{equation}
Define $ \v p = [p_1, p_2, \d, p_K]^\transpose$. 
We still need to find $ \m Q $ and $ \v p $ such that Eqs.~\eqref{equ:slack_2} and~\eqref{equ:pri_const1}--\eqref{equ:pri_var2} hold. 

\subsubsection{Solving Eq. \eqref{equ:pri_const1} for $ \v p $}
Substituting $\hat{\m V}_k = \hat{\v v}_k \hat{\v v}_k^\hermitian$ into \eqref{equ:pri_const1}, one has  
\begin{equation*}
	\frac{1}{\bar \gamma_k} p_k \v h_k^\hermitian \hat{\m V}_k \v h_k  -  \sum_{j\neq k} p_j \v h_k^\hermitian \hat{\m V}_j \v h_k - \v h_k^\hermitian \m Q \v h_k  - \sigma_k^2= 0.
\end{equation*}
Then one can solve for $p_k$ as follows:
\begin{equation}
	p_k\left(\m Q, \v p\right) = \frac{\bar \gamma_k \left(\sum_{j\neq k} p_j \v h_k^\hermitian \hat{\m V}_j \v h_k  + \v h_k^\hermitian \m Q \v h_k + \sigma_k^2\right)}{ \v h_k^\hermitian \hat{\m V}_k \v h_k }. 
\label{equ:Q2p}
\end{equation}

\begin{figure}[!t]
	\centering
	\includegraphics[width=0.45\textwidth]{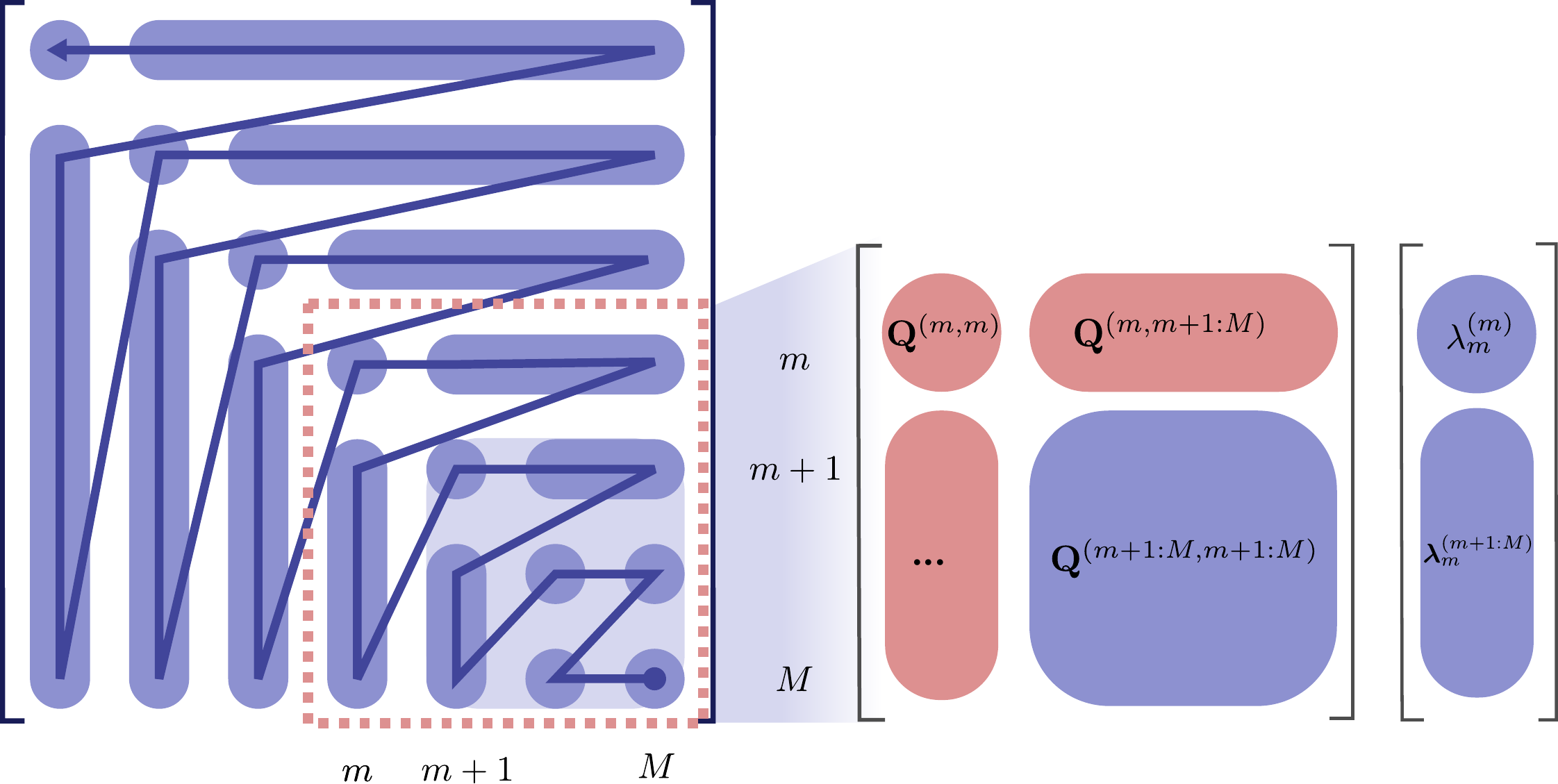}
	\vspace{-10pt}
	\caption{An illustration of solving \eqref{equ:equ_Bm} for $ \m Q $. }
	\label{fig:illustration_solveQ}
\end{figure}

\subsubsection{Solving Eqs.~\eqref{equ:slack_2}, \eqref{equ:pri_const2}, and \eqref{equ:pri_var2} for $ \m
Q $}
Next, given $ \v p $, we shall obtain $ \m Q $ such that Eqs.~\eqref{equ:slack_2}, \eqref{equ:pri_const2}, and \eqref{equ:pri_var2} hold. 
By Eq. \eqref{equ:dual_var2}, one can decompose $\m{\Lambda}_m$ into
$\m{\Lambda}_m = \v{\lambda}_m \v{\lambda}_m^\hermitian,$
where $\v{\lambda}_m = \left[\v{0}, \v{\lambda}_m^{(m)}, \v{\lambda}_m^{(m+1)}, \d, \v{\lambda}_m^{(M)}\right]^\transpose$. 
This decomposition, together with Eqs. \eqref{equ:pri_const2} and \eqref{equ:slack_2}, implies
\begin{equation}
	\m B_m\left(\left\{p_k \hat{\m V}_k\right\},\m Q\right) \v{\lambda}_m = \m 0 , \fm{}.
	\label{equ:equ_Bm}
\end{equation}

We shall solve \eqref{equ:equ_Bm} from $m=M$ to $m=1$ and obtain the desired $\m Q$ in the order shown in the left-hand side of Fig.~\ref{fig:illustration_solveQ}. 
More specifically, when $m = M$, it follows that
\begin{equation}
	\m Q^{(M,M)} = \frac{\sum_{k\in\cK} p_k \hat{\m V}_k^{(M,M)}}{2^{\bar C_m}- 1}. 
	\label{equ:QMM}
\end{equation}
When $ m < M $, we can substitute the known $ \m Q^{(m+1:M,m+1:M)} $ into the $m$-th equation in Eq. \eqref{equ:equ_Bm}, which gives a linear equation with variables $ \m Q^{(m,m)}, \m Q^{(m, m+1:M)}$, and $ \m Q^{(m+1:M, m)} $. 
This linear equation is illustrated in the right-hand side of Fig.~\ref{fig:illustration_solveQ}, where the known variables are marked in blue and the unknown variables are marked in red. 
We first solve $\m Q^{(m+1:M, m)}$ by using the last $M-m+1$ equations in the right-hand side of Fig.~\ref{fig:illustration_solveQ}. 
The solution is given by 
\begin{equation}
	\m Q^{(m+1:M, m)} = - \frac{\m Q^{(m+1:M, m+1:M)} \v{\lambda}_m^{(m+1:M)}}{\v{\lambda}_m^{(m)}}. 
	\label{equ:solve_Q1}
\end{equation}
Then, by Eq. \eqref{equ:pri_var2}, $ \m Q^{(m, m+1:M)} $ is given by the Hermitian transpose of $ \m Q^{(m+1:M,m)} $. 
Finally, we can further obtain $\m Q^{(m,m)}$ by using the first equation in the right-hand side of Fig.~\ref{fig:illustration_solveQ}, which is given by 
\begin{equation}
    \begin{aligned}
        	\m Q^{(m,m)} = \frac{2^{\bar C_m}}{2^{\bar C_m} - 1} \frac{(\v{\lambda}_m^{(m+1:M)})^{\hermitian} \m Q^{(m+1:M, m+1:M)} \v{\lambda}_m^{(m+1:M)}}{|\v \lambda_m^{(m)}|^2} \\ +~ \frac{1}{2^{\bar C_m} - 1} \sum_{k\in\cK} p_k \hat{\m V}_k^{(m,m)}. 
    \end{aligned}\label{equ:solve_Q3}
\end{equation}
Using the above tricks, we can obtain the solution $\m Q$ of Eq.~\eqref{equ:equ_Bm}. We denote the solution as $\m Q(\v p),$ because the solution depends on the given $\v p$. 

\subsubsection{Primal Fixed Point Iteration}
Based on the above discussion, we know that 
if $\v p$ is known, one can get $\m Q(\v p)$ such that Eqs.~\eqref{equ:slack_2}, \eqref{equ:pri_const2}, and \eqref{equ:pri_var2} hold. 
Plugging this solution into \eqref{equ:Q2p} gives 
\begin{equation}
	p_k = J_k(\v p) \triangleq p_k\left(\m Q(\v p), \v p\right), \fk. 
	\label{equ:raw_primal_iter}
\end{equation}
Define $J(\v p) = \left[J_1(\v p), J_2(\v p), \d, J_K(\v p)\right]^\transpose$. Then Eq. \eqref{equ:raw_primal_iter} becomes the problem of finding the fixed point of the function $ J(\cdot) $, i.e.,  
\begin{equation}
	\v p = J(\v p). 
	\label{equ:primal_fix_point}
\end{equation}
If one can find $ \v p $ such that \eqref{equ:primal_fix_point} holds, then $\v p$ and $\m Q(\v p)$ will satisfy Eqs.~\eqref{equ:slack_2}~and~\eqref{equ:pri_const1}--\eqref{equ:pri_var2}, and further $ \left\{ p_k \hat{\m V}_k \right\} $ and $ \m Q(\v p) $ will solve the enhanced primal feasibility conditions. 

We shall show in Theorem \ref{thm:alg_prop} that Eq. \eqref{equ:primal_fix_point} can be efficiently solved via the fixed point iteration 
\begin{equation}
	\v p^{(i+1)} = J(\v p^{(i)}), 
	\label{equ:pri_iter}
\end{equation}
and the convergence rate of the fixed point iteration is linear. 
Moreover, each evaluation of $ J(\cdot) $ is computationally cheap. 
More specifically, the computation mainly consists of two parts: 
first, the total complexity of the procedure described by \eqref{equ:QMM}--\eqref{equ:solve_Q3} is $ \mathcal{O}(M^3) $;  
second, after $ \v  h_k^\hermitian \hat{\m V}_j \v h_k $ for all $ j $ and $ k $ are computed, the complexity of computing $ p_k(\m Q, \v p) $ in \eqref{equ:Q2p} is $ \mathcal{O}(KM^2) $. 
As a result, the total complexity of each evaluation of $J(\cdot)$ is $ \mathcal{O}(M^2\max\{K,M\}) $. 
Due to the low per-iteration complexity and the linear convergence rate, the above fixed point iteration in \eqref{equ:pri_iter} provides an efficient way of solving the enhanced primal feasibility conditions after solving the dual problem. 

\subsection{Proposed Fixed Point Iteration (FPI) Algorithm} \label{ss:alg}

Now, we present the algorithm for solving problem~\eqref{SDR} (which is equivalent to problem (P) by Theorem \ref{thm:tight}). 
The algorithm first finds $\v \beta$ and $\Lm$ that satisfy the enhanced dual feasibility conditions; 
with found $\v \beta$ and $\Lm$ fixed, the algorithm then finds $\Vk$ and $\m Q$ that satisfy the rest of the enhanced KKT conditions. 
Hence, $\Vk$, $\m Q$, $\v \beta,$ and $\Lm$ together satisfy the enhanced KKT conditions and thus is a KKT point of problem~\eqref{SDR}. 
Since $\rk\left(\m V_k\right) = 1$ for all $k$, we can recover the optimal solution for problem (P). 
The pseudocodes of the proposed FPI algorithm are given in \algref{alg:the_alg}. 

\begin{algorithm}[H]
	\caption{Proposed FPI Algorithm for Solving Problem (P)}
	\begin{algorithmic}[1]
		\STATE Find $\v \beta$ and $\Lm$ that satisfy the enhanced dual feasibility conditions by performing the fixed point iteration in \eqref{equ:dual_fix_point} on $\v \beta$ until the desired error bound is met. 
		\STATE Find $\Vk$ and $\m Q$ that satisfy the rest of the enhanced KKT conditions by performing the fixed point iteration in \eqref{equ:primal_fix_point} on $\v p$ until the desired error bound is met. 
		\STATE Find $\v v_k$ such that $\m V_k = \v v_k \v v_k^\hermitian, \fk$. 
		\STATE \textbf{Output:} $\{\v v_k\}$~and~$\m Q.$
	\end{algorithmic}
	\label{alg:the_alg}
\end{algorithm}

While designed for solving the joint beamforming and compression problem, the proposed FPI algorithm can be extended to handle the joint beamforming and compression problem with additional constraints, e.g., per-antenna power constraints \cite{fan2024JointBeamformingCompression}.

\subsection{Theoretical Guarantees}
\label{ss:th_p_a}

\subsubsection{Convergence and Convergence Rate of FPI} We have the following convergence and convergence rate guarantee of the proposed FPI algorithm. 
\begin{theorem}
	\label{thm:alg_prop}
	If problem~\eqref{SDR} is strictly feasible, both the dual fixed point iteration \eqref{equ:dual_iter} and the primal fixed point iteration \eqref{equ:pri_iter} in the proposed FPI algorithm converge linearly. 
\end{theorem}
\begin{proof}
	See Appendix~\ref{apd:alg_prop}. 
\end{proof} 

The following convergence rate results have been shown for the dual fixed point iteration in Appendix~\ref{apd:alg_prop}. 
Let $\tilde{\v \beta} = [\tilde \beta_1, \tilde \beta_2, \d, \tilde \beta_K]^\transpose$. 
Define the metric $ \mu: \mathbb{R}_{++}^K \times \mathbb{R}_{++}^K \rightarrow \mathbb{R}_+^K $ as 
\begin{equation}
	\mu(\v \beta, \tilde{\v \beta}) = \max_{k \in \cK} \left| \log_{\mathrm{e}} \left( \frac{\beta_k}{\tilde \beta_k} \right) \right|,
	\label{equ:metric}
\end{equation}
which is proposed in \cite{nuzman2007ContractionApproachPower}. 
Under this metric, the asymptotic linear convergence rate of the dual fixed point iteration \eqref{equ:dual_iter} is given by 
\begin{equation}
	\limsup_{i\rightarrow \infty}\frac{\mu(\v \beta^{(i+1)}, \v \beta^*)}{\mu(\v \beta^{(i)}, \v \beta^*)} \leq \frac{\lambda(\v \beta^*)}{1 + \lambda(\v \beta^*)}, 
	\label{equ:c_r_limit}
\end{equation}
where $\v \beta^*$ is the fixed point of $I(\cdot)$ and 
\begin{equation}
	\lambda(\v \beta) = \max_{k \in \cK} \left\{ \left\|\m C_k(\v \beta, \left\{ \m \Lambda_m(\v \beta) \right\}) - \m I \right\|_2\right\}. 
	\label{equ:def_lambda}
\end{equation}
Note that the convergence rate in \eqref{equ:c_r_limit} is established within Thompson's metric space, which differs from the conventional notion of the convergence rate in the Euclidean space. 
However, as demonstrated in \cite[Proposition 2]{piotrowski2022FixedPointIteration}, geometric convergence in Thompson's metric space implies geometric convergence in Euclidean settings. This result enables the interpretation of the convergence behavior across these two distinct metric spaces.
Besides, the linear convergence rate of the primal fixed point iteration \eqref{equ:pri_iter} is governed by the spectral radius of $ \m G $ in $ J(\cdot) $, where the entries of $ \m G $ are given by 

\begin{equation}
	G_{kj} = \left\{
	\begin{aligned}
		&\frac{\bar \gamma_k \v h_k^\hermitian \m Q(\v{e}_k) \v h_k }{\v h_k^\hermitian \hat{\m V}_k \v h_k}, & \text{ if } j=k, \\
		&\frac{\bar \gamma_k \left(\v h_k^\hermitian \hat{\m V}_j \v h_k + \v h_k^\hermitian \m Q(\v{e}_k) \v h_k \right)}{\v h_k^\hermitian \hat{\m V}_k \v h_k}, &\text{ otherwise}. 
	\end{aligned}
	\right.
	\label{equ:entry_G}
\end{equation}

The above convergence rate results shed useful insights into the convergence  behavior of the proposed FPI algorithm, and in particular the efficiency of the proposed FPI algorithm is determined by the given problem instance. 
In general, the proposed FPI algorithm will converge slower when the considered problem approaches the singular boundary\footnote{The singular boundary here means the boundary of the achievable SINR region by problem~\eqref{SDR} with given channel conditions and compression capacities. }. 
To be specific, as we increase the SINR targets $\{\bar \gamma_k\}$ in problem~\eqref{SDR} (with all the other parameters being unchanged), the corresponding problem will approach the singular boundary. 
In the dual fixed point iteration \eqref{equ:dual_iter}, as $\{\bar \gamma_k\}$ increases, $\v\beta^*$, which depends on $\{\bar \gamma_k\}$, will increase.
Combining this with the fact that $\lambda(\v \beta)$ in \eqref{equ:def_lambda} is an increasing function of $\v\beta$, we have that $\lambda(\v \beta^*)$ will increase, and hence the right-hand side of \eqref{equ:c_r_limit} will increase as well. 
This suggests that the convergence rate of the dual fixed point iteration \eqref{equ:dual_iter} becomes slower as the problem approaches the singular boundary.
The same happens for the primal fixed point iteration \eqref{equ:pri_iter}. 
When the SINR targets $\{\bar \gamma_k\}$ in problem~\eqref{SDR} increase, the spectral radius of $\m G$ in \eqref{equ:entry_G} increases, and the convergence rate of the primal fixed point iteration \eqref{equ:primal_fix_point} becomes slower. 
This shows that the primal fixed point iteration \eqref{equ:pri_iter} will become slower as the problem approaches the singular boundary. 

\subsubsection{Global Optimality of FPI and Infeasibility Detection} 
First, we have the following global optimality guarantee of the proposed FPI algorithm. 
\begin{theorem}
	If problem~\eqref{SDR} is strictly feasible, then the FPI algorithm returns the optimal solution of problem (P). 
\label{thm:global}
\end{theorem}
\begin{proof}
	Let $ \v p $ be the converged solution of the primal fixed point iteration \eqref{equ:pri_iter}, and $ \v \beta $ be the converged solution of the dual fixed point iteration \eqref{equ:dual_iter} in the FPI algorithm. 
	First, for given $ \v p $, $ \m Q(\v p) $ in the FPI algorithm is obtained such that Eqs. \eqref{equ:pri_const2}--\eqref{equ:pri_var2} hold; 
	for given $ \v \beta $, $ \{ \m \Lambda_m(\v \beta) \} $ in the FPI algorithm are obtained such that Eqs. \eqref{equ:dual_const2} and \eqref{equ:dual_var2} hold. 
	Second, $ \{ \hat{\m V}_k \} $ are solved such that Eqs. \eqref{equ:slack_1} and \eqref{equ:pri_var1} hold. 
	Furthermore, the convergence of $ \v \beta $ and $ \v p $ shows that Eqs. \eqref{equ:dual_const1}, \eqref{equ:dual_var1}, and \eqref{equ:pri_const1} hold. 
	In all, the primal variables $ \{ p_k \hat{\m V}_k \}$ and $\m Q(\v p)$, together with the dual variables $\v \beta$ and $\{\m \Lambda_m(\v \beta)\}$, satisfy all conditions \eqsta{} in the enhanced KKT system, and hence is a global solution of problem~\eqref{SDR} (which is equivalent to problem (P)). 
	
\end{proof}

In Theorem 3, we assume that problem~\eqref{SDR} is feasible. 
A natural question is how the FPI algorithm behaves when problem~\eqref{SDR} is infeasible. 
According to the weak duality, any dual feasible solution provides a lower bound on the optimal value of problem~\eqref{SDR}. 
By the monotonicity of $I(\cdot)$ shown in Appendix~\ref{apd:alg_prop}, if we initialize the dual fixed point iteration \eqref{equ:dual_iter} with $ \v \beta^{(0)} $ satisfying $ \v \beta^{(0)} \leq I(\v \beta^{(0)})$ and $ \v \beta^{(0)} \neq I(\v \beta^{(0)}) $, e.g., $\v \beta^{(0)} = \v{0} $, we have 
\begin{equation*}
	I(\v \beta^{(i)}) = \v \beta^{(i+1)} \geq \v \beta^{(i)} \text{~and~} \v \beta^{(i+1)} \neq \v \beta^{(i)}, \text{~for~all~} i. 
\end{equation*}
This shows that $(\v \beta^{(i)},  \{\m \Lambda_m(\v \beta^{(i)})\} )$ is a dual feasible solution, and hence the corresponding strictly increasing dual objective value serves as a lower bound on the optimal value of problem~\eqref{SDR}. 
If we observe that this dual objective value is greater than a preset upper bound in practice (e.g., the system power limit), then we claim that problem~\eqref{SDR} is infeasible. 
In this sense, our proposed FPI algorithm can automatically detect the infeasibility of problem~\eqref{SDR} when the problem is infeasible. 
 
In summary, two key features of our proposed FPI algorithm are as follows: 
(1) it is guaranteed to find the global solution of problem~\eqref{SDR} when the problem is feasible and is able to detect its infeasibility when the problem is infeasible (with an appropriate initialization); 
(2) it enjoys a linear convergence rate, and it converges faster when the problem instance is far away from the singular boundary (compared with the problem instance close to the singular boundary). 

\section{Numerical Results}

In this section, the performance of the proposed FPI algorithm is evaluated and compared with existing SOTA algorithms. 
Unless otherwise specified, we consider the following default scenario. 
We consider a downlink cooperative cellular network with $M = 7$ single-antenna relays serving $K = 8$ single-antenna users, where the relays are located at the center of the wrapped-around hexagonal cells \cite{shen2018FractionalProgrammingCommunication} and the users are randomly placed.
The distance between two neighbouring relays is $150$~m and the height of the relays is $30$~m. 
The channel between the users and the relays are generated by the Rayleigh fading model. 
Following \cite{liu2017FronthaulAwareDesignCloud}, we set the pathloss as $140.7+36.7 \log_{10}(d)$ dB, where $d$ is the distance in kilometer, the noise power spectral density as $-169$ dBm/Hz, and the total system bandwidth as $20$ MHz. 
In the following numerical experiments, all the users' SINR targets, denoted by $\bar \gamma = 4$ dB, and all the fronthaul capacities, denoted by $\bar C = 3$, are considered identical, respectively.  
These parameter settings ensure that the resulting problem is feasible. 
Furthermore, we initialize the dual fixed point iteration \eqref{equ:dual_iter} and the primal fixed point iteration \eqref{equ:pri_iter} with zero vectors in the proposed FPI algorithm. 

\subsection{Behaviors of Proposed FPI Algorithm}

In this subsection, we first study the practical behaviors of the proposed FPI algorithm. 
We are interested in the convergence rate of the proposed FPI algorithm and how it behaves when the corresponding problem instances are close to the singular boundary. 

\begin{figure}[t]
	\centering
	\subfloat[]{\includegraphics[width=0.23\textwidth]{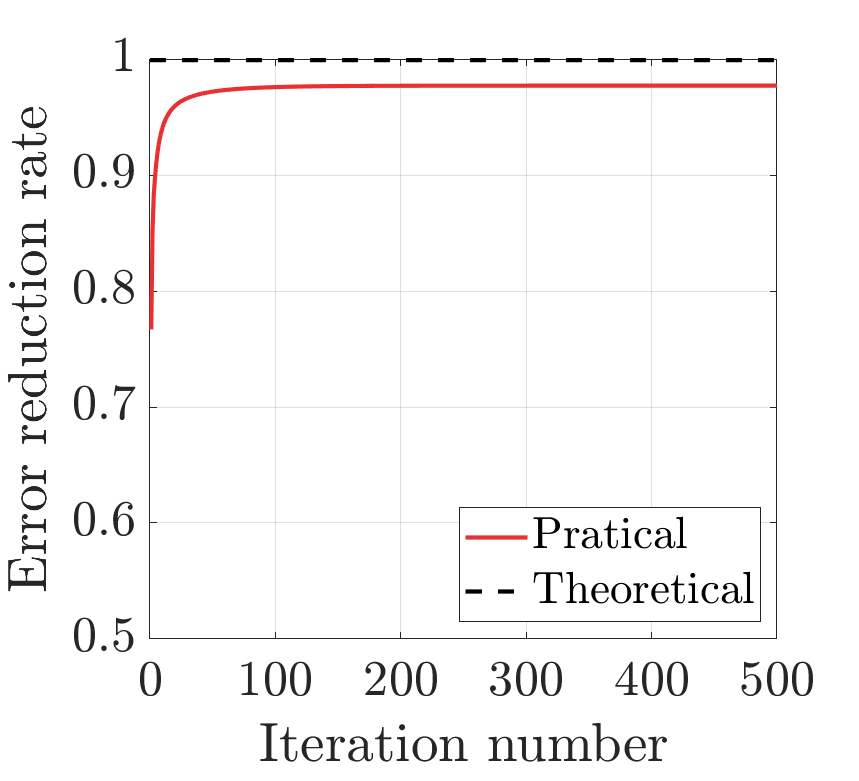}
	\label{fig:b1_dual}}
	\subfloat[]{\includegraphics[width=0.23\textwidth]{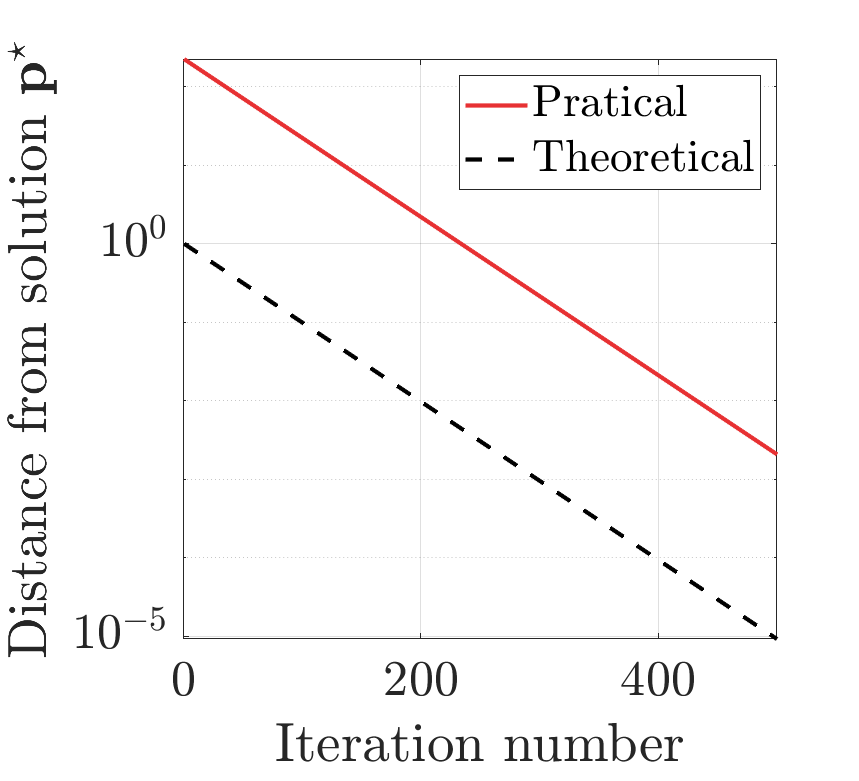}
	\label{fig:b1_pri}}\\
	\caption{(a) The convergence rate of the dual fixed point iteration \eqref{equ:dual_iter} with $\bar \gamma = 4$ dB; 
	(b) The convergence rate of the primal fixed point iteration \eqref{equ:pri_iter} with $\bar \gamma = 4$ dB. }
	\label{fig:b1}
\end{figure}

Fig.~\ref{fig:b1} shows the linear convergence rate of the proposed FPI algorithm. 
In particular, Fig.~\ref{fig:b1}~\subref{fig:b1_dual} shows the linear convergence rate of the dual fixed point iteration \eqref{equ:dual_iter}. 
In Fig.~\ref{fig:b1}~\subref{fig:b1_dual}, the error reduction rate $ \frac{\mu(\v \beta^{(i+1)}, \v \beta^*)}{\mu(\v \beta^{(i)}, \v \beta^*)} $ is plotted against the iteration number $ i $ and compared with the theoretical upper bound $\frac{\lambda(\v \beta^*)}{1 + \lambda(\v \beta^*)}$ given in \eqref{equ:c_r_limit} for a given SINR target $\bar \gamma = 4 $~dB. 
We can observe from the figure that the convergence rate of the dual fixed point iteration \eqref{equ:dual_iter} is indeed linear, albeit the upper bound in \eqref{equ:c_r_limit} is conservative and does not match the practical convergence rate. 
Fig.~\ref{fig:b1}~\subref{fig:b1_pri} verifies the linear convergence rate of the primal fixed point iteration, and the theoretical and practical convergence rates are well-matched.

Table.~\ref{tb:1} shows the practical asymptotic linear convergence rate and the theoretical upper bound $ \frac{\lambda(\v \beta^*)}{1 + \lambda(\v \beta^*)} $ given in \eqref{equ:c_r_limit} for different SINR targets $\bar \gamma$. 
As is shown in the table, when the SINR target $\bar \gamma$ increases (i.e. the feasible problem approaches the singular boundary), the theoretical and practical asymptotic convergence rates become close to each other, and both of them increase to one. 
This verifies the convergence behavior analysis below \eqref{equ:entry_G} and shows that the theoretical upper bound $\frac{\lambda(\v \beta^*)}{1 + \lambda(\v \beta^*)}$ is useful in characterizing the (intrinsic) difficulty of the problem. 

\begin{table}[!t]
    \centering 
	\caption{Theoretical and practical convergence rate comparison for different SINR targets $\bar \gamma$}
    \label{tb:1}
    \begin{tabular}{llllll}
    \hline
    SINR target $\bar \gamma$ (dB) & 3.6  & 3.7  & 3.8   & 3.9   & 4.0  \\ \hline
    Theoretical     & 0.995 & 0.996 & 0.997 & 0.998 & 0.999 \\
    Practical     & 0.898 & 0.916 & 0.935 & 0.956 & 0.977 \\ \hline
    \end{tabular}
\end{table}

Now we look at the behaviors of the proposed FPI algorithm when the problem instances are close to the singular boundary. 
Fig.~\ref{fig:b2} shows the behavior of the FPI algorithm in both feasible and infeasible cases. 
In both Figs.~\ref{fig:b2}~\subref{fig:b2_fea}~and~\ref{fig:b2}~\subref{fig:b2_infea}, three problem instances with different SINR targets $ \bar \gamma$ are plotted. 
Figs.~\ref{fig:b2}~\subref{fig:b2_fea} plots the dual/primal objective values versus the iteration number when the problems are feasible. 
The optimal values of the three problem instances are plotted in the gray dotted line, respectively in Fig.~\ref{fig:b2}~\subref{fig:b2_fea}. 
As expected (and observed from Fig.~\ref{fig:b2}~\subref{fig:b2_fea}), both the dual and primal objective values monotonically increase and finally converge to the optimal value. 
Fig.~\ref{fig:b2}~\subref{fig:b2_infea} plots the dual objective values versus the iteration number when the problems are infeasible. 
In this case, as analyzed below Theorem \ref{thm:global} and observed in Fig.~\ref{fig:b2}~\subref{fig:b2_infea}, the dual objective values monotonically increase to infinity. 
It can be seen clearly from Fig.~\ref{fig:b2} that the proposed FPI algorithm converges/diverges slower when the problem instances are close to the singular boundary in both feasible and infeasible cases. 

\begin{figure}[t]
	\centering
	\subfloat[]{\includegraphics[width=0.23\textwidth]{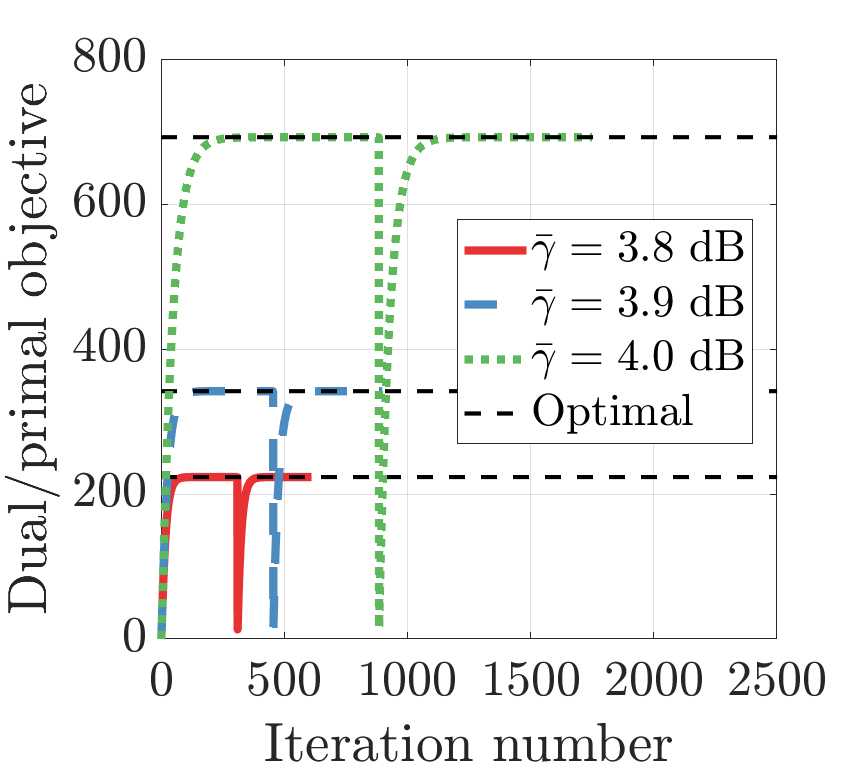}
	\label{fig:b2_fea}}
	\subfloat[]{\includegraphics[width=0.23\textwidth]{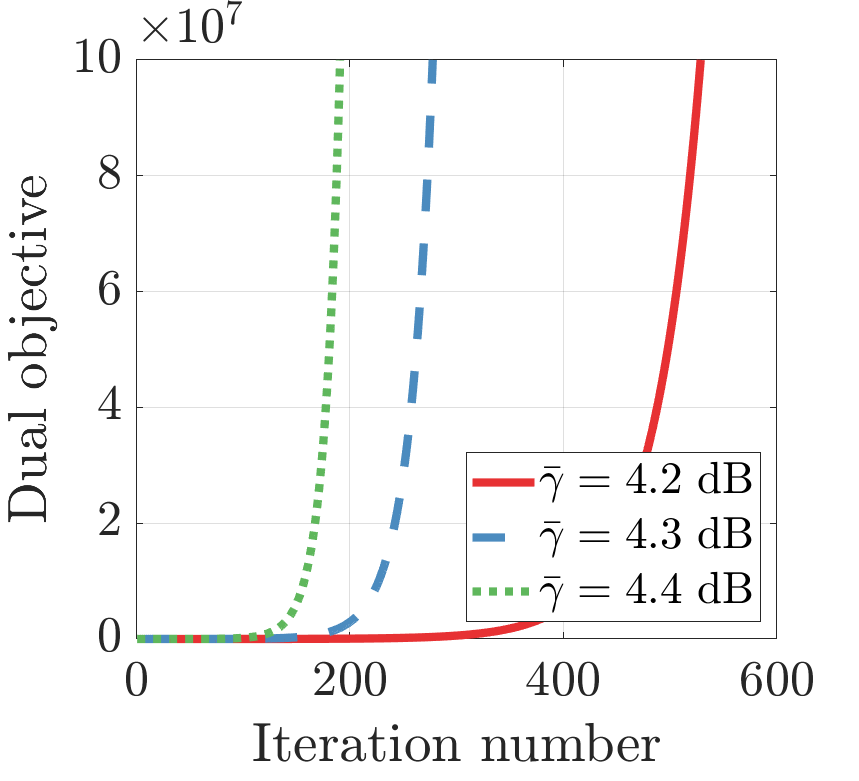}
	\label{fig:b2_infea}}\\
	\caption{(a) The dual/primal objective values versus the iteration number for feasible problem instances with different SINR targets $ \bar \gamma $; 
	(b) The dual objective values versus the iteration number for infeasible problem instances with different SINR targets $ \bar \gamma $. }
	\label{fig:b2}
\end{figure}
 
\subsection{Comparison with SOTA Algorithms}

In this subsection, to illustrate the efficiency of the proposed FPI algorithm, we compare it with the following three SOTA benchmarks: 
\begin{itemize}
	\item \emph{SDR}: We call CVX \cite{CVX} to directly solve the SDR in \eqref{SDR}. 
	This benchmark is helpful in verifying the tightness of the corresponding SDR (i.e., Theorem~\ref{thm:tight}) as well as the global optimality (i.e., Theorem~\ref{thm:global}). 

	\item \emph{UD} \cite{liu2021UplinkdownlinkDualityMultipleaccess}: 
	The UD algorithm first uses a fixed point iteration to solve the dual uplink problem (which is obtained by transforming the Lagrangian dual problem of problem (P) with fixed beamformers); then calls CVX to solve the reduced downlink problem with fixed beamformers (which is convex). The UD algorithm is also guaranteed to find the global solution of problem (P). 

	\item \emph{SCA}\cite{sun2017MajorizationminimizationAlgorithmsSignal}: The SCA algorithm solves problem (P) by iteratively solving a sequence of convex approximation subproblems, and each convex approximation subproblem is obtained by linearizing the SINR constraints in problem (P) at the current point. 
	This benchmark shows the performance when no structure of the problem is exploited. 
	Therefore, it is useful to compare the performance gain by utilizing the special structure of the problem. 
\end{itemize}

\begin{figure}[!t]
	\centering
	\subfloat[]{\includegraphics[width=0.4\textwidth]{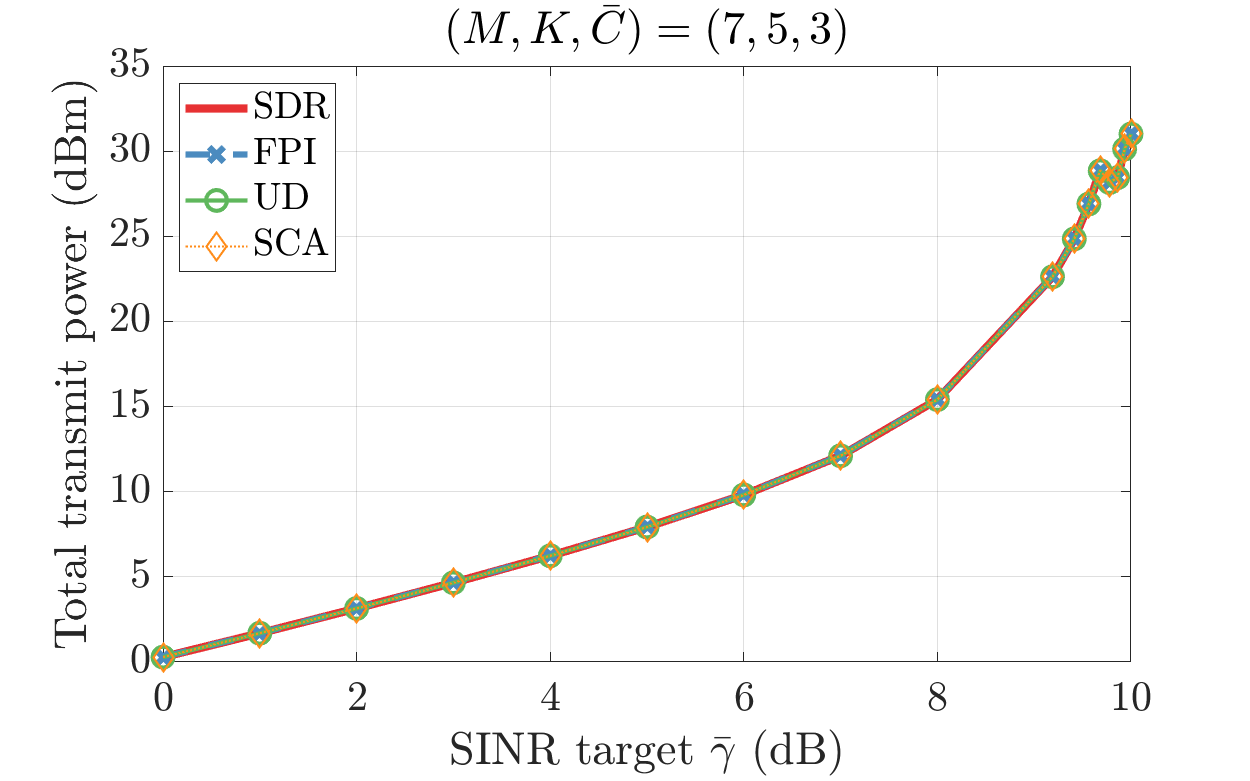}
		\label{fig:753}}\\
	\subfloat[]{\includegraphics[width=0.4\textwidth]{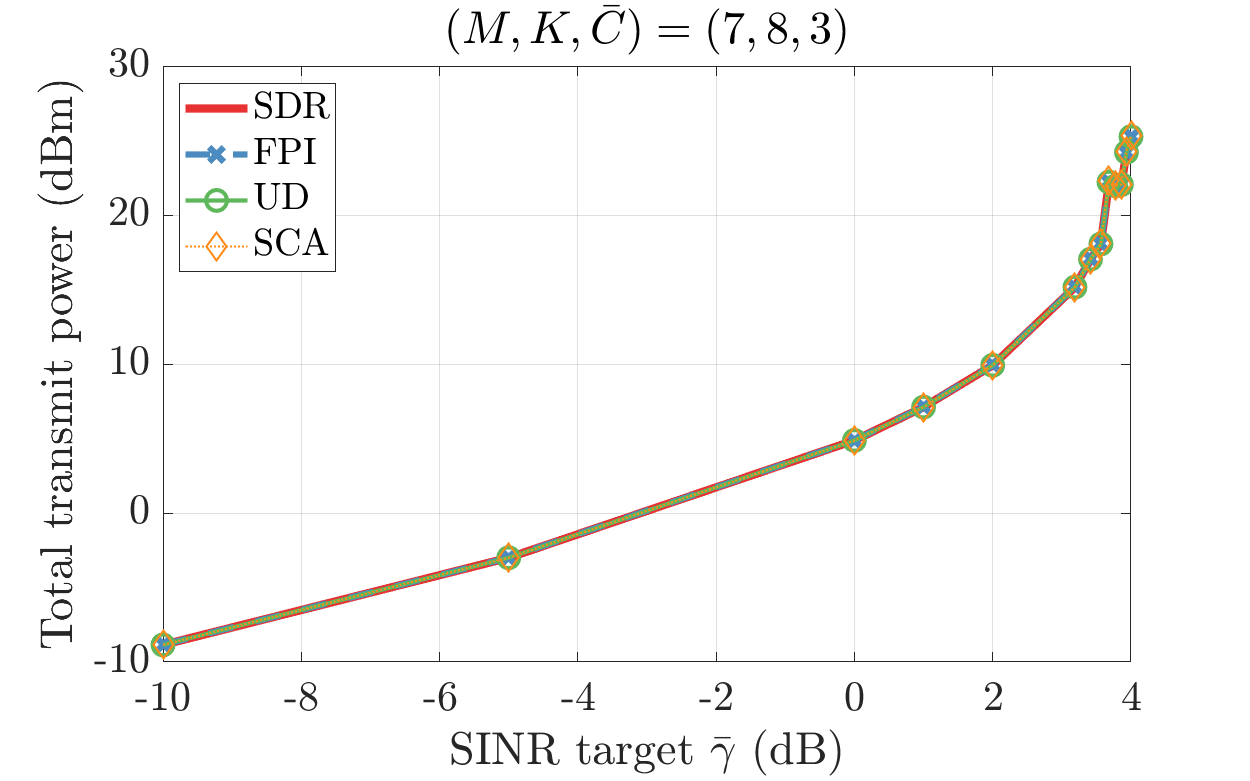}
		\label{fig:783}}\\
	\subfloat[]{\includegraphics[width=0.4\textwidth]{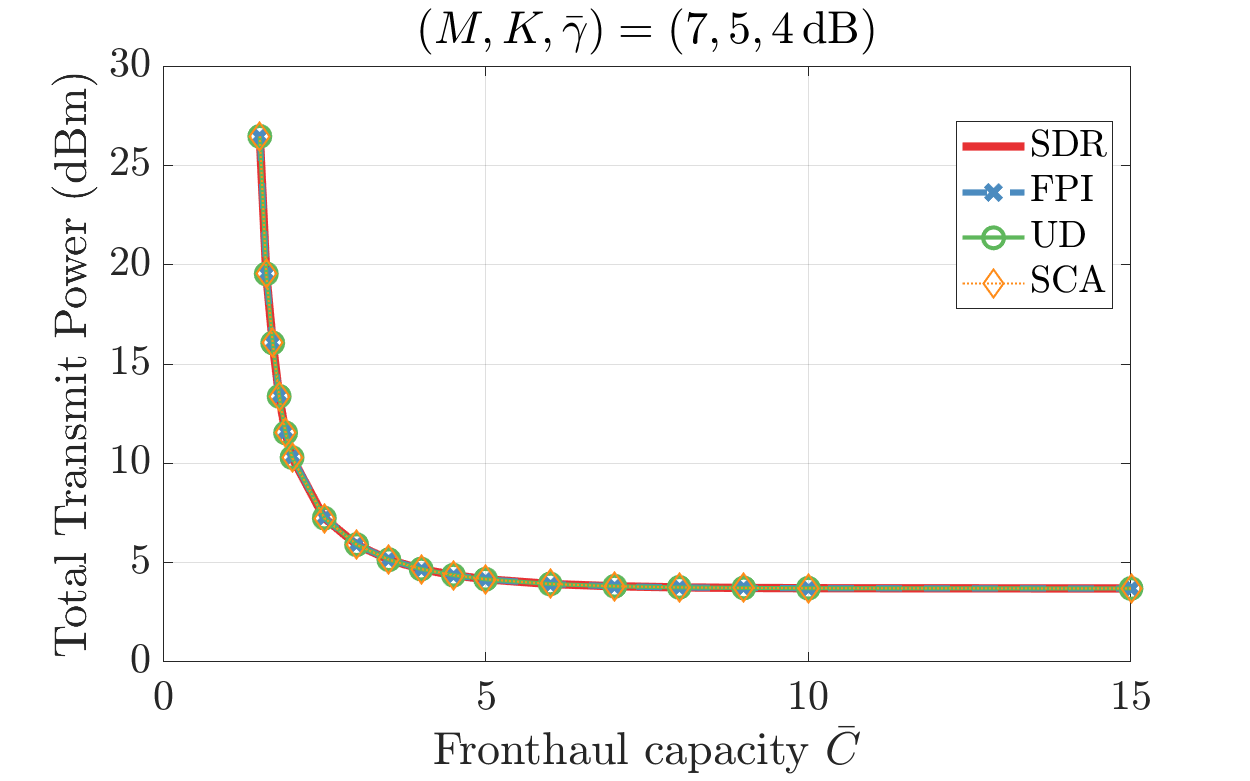}
		\label{fig:754}}\\
	\subfloat[]{\includegraphics[width=0.4\textwidth]{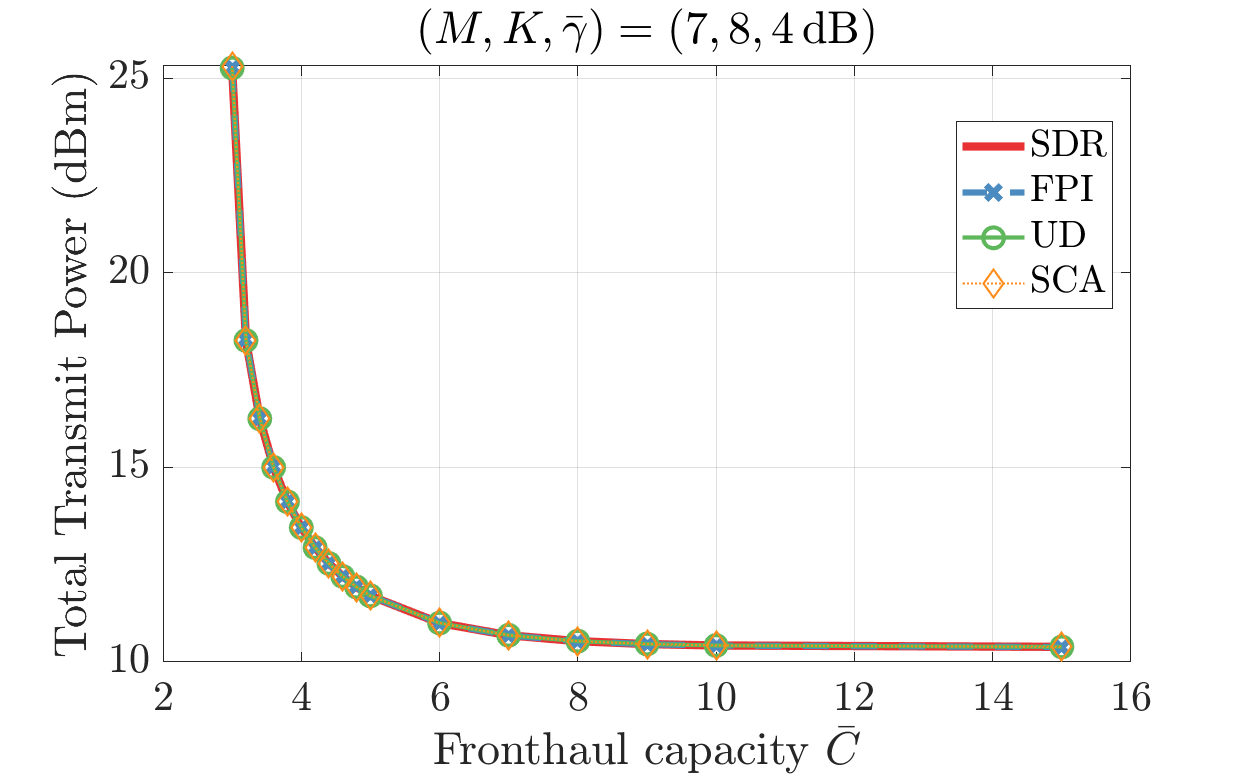}
		\label{fig:784}}\\
	\caption{The average objective value (i.e., the total transmit power) of different algorithms with different system parameters $(M, K, \bar \gamma, \bar C)$. }
	\label{fig:SINR}
\end{figure}

Figs.~\ref{fig:SINR}~and~\ref{fig:SINR_T} show the performance comparison of the proposed FPI algorithm and the three benchmarks. 
In the following numerical experiments, each data point is obtained by averaging over $ 200 $ channel realizations. 
Fig.~\ref{fig:SINR} plots the average objective values at the solutions obtained by four different algorithms. 
We can see from Fig.~\ref{fig:SINR} that all the four algorithms return the same solution. 
This verifies the tightness result of the SDR (i.e., Theorem \ref{thm:tight}) and the global optimality of the solution returned by the proposed algorithm (i.e., Theorem~\ref{thm:global}). 

Fig.~\ref{fig:SINR_T} plots the average CPU time taken by different algorithms. 
From Fig.~\ref{fig:SINR_T}, we can observe that the SCA algorithm performs the worst (even though they are initialized with a point that is close to the optimal solution). 
The SDR algorithm generally performs better than the SCA algorithm in terms of the CPU time. 
It is also observed that the CPU time of the UD algorithm and the FPI algorithm both increases as the problem approaches the singular boundary (by increasing the SINR target and/or decreasing the fronthaul capacities), which is consistent with our analysis in Section IV-E. 
In all, the UD algorithm has a similar performance as that of the SDR algorithm, and the proposed FPI algorithm significantly outperforms the other algorithms in terms of the CPU time. 
Figs.~\ref{fig:SINR}~and~\ref{fig:SINR_T} together show the global optimality and high efficiency of our proposed FPI algorithm.

\begin{figure}[!t]
	\centering
	\subfloat[]{\includegraphics[width=0.4\textwidth]{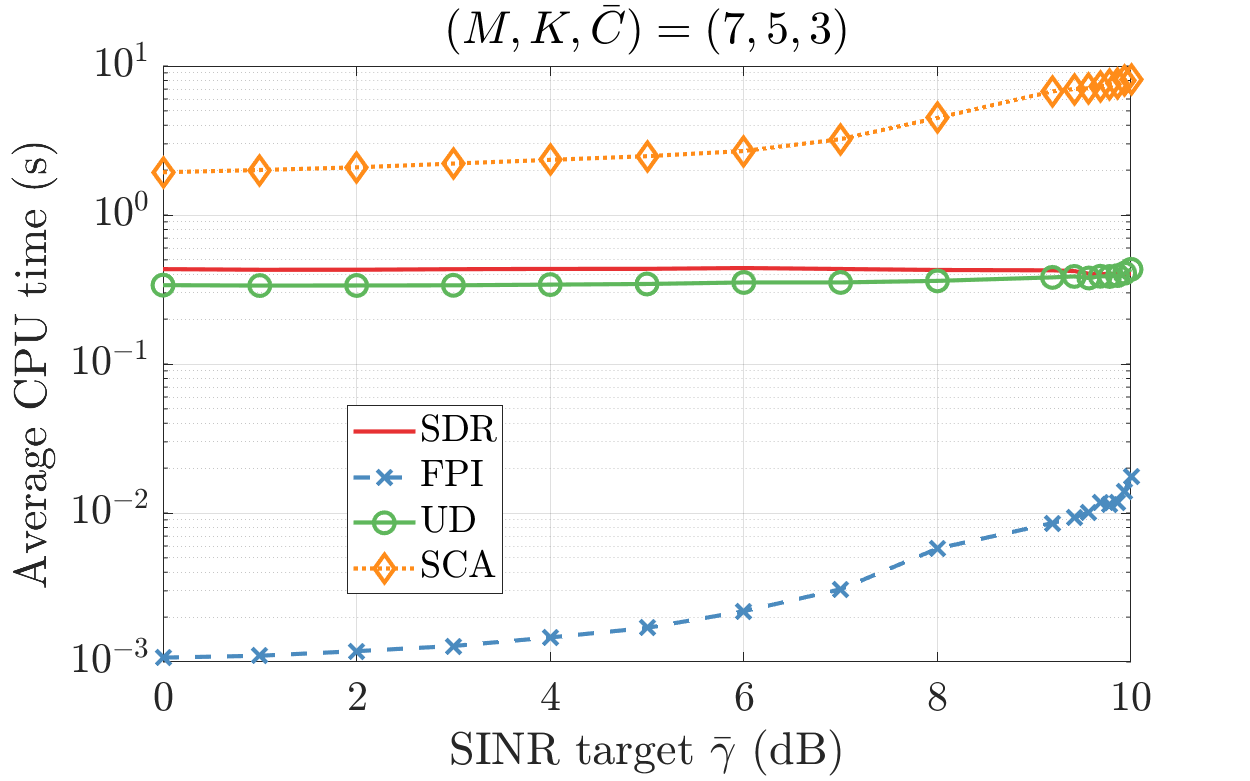}
	\label{fig:753_T}}\\
	\subfloat[]{\includegraphics[width=0.4\textwidth]{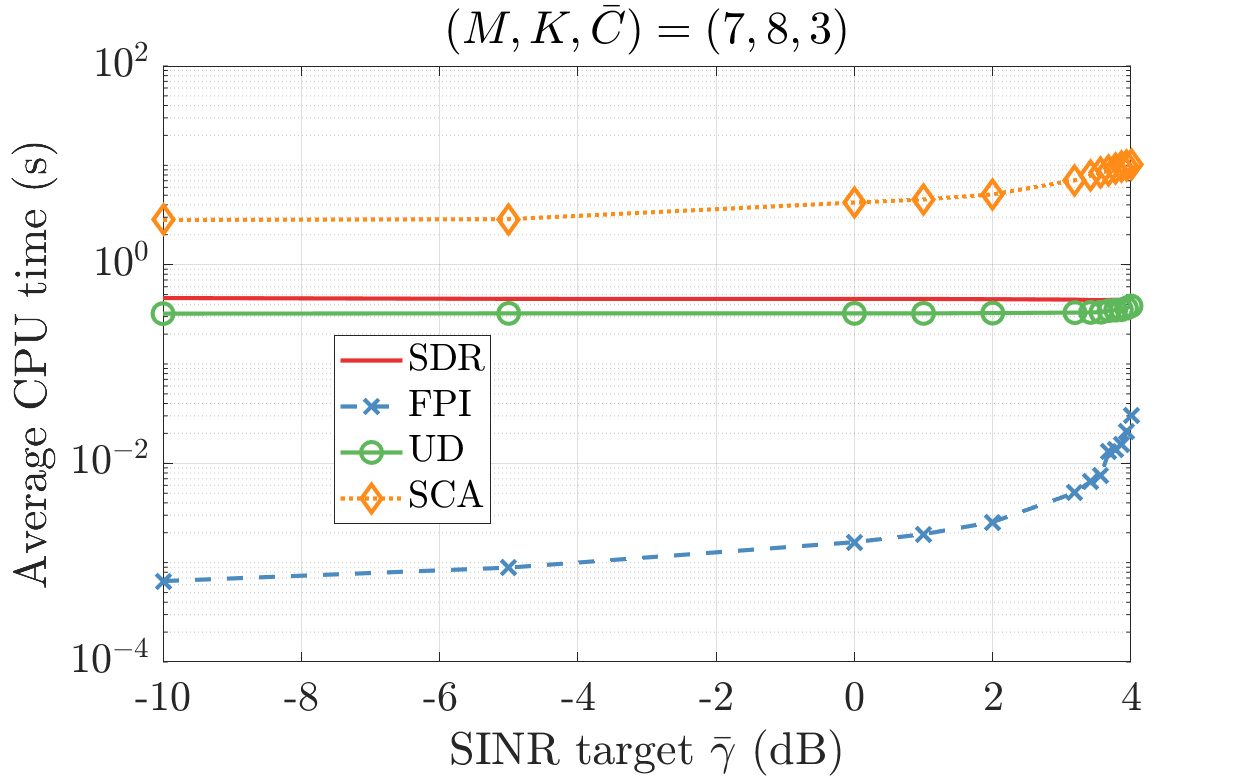}
	\label{fig:783_T}}\\
	\subfloat[]{\includegraphics[width=0.4\textwidth]{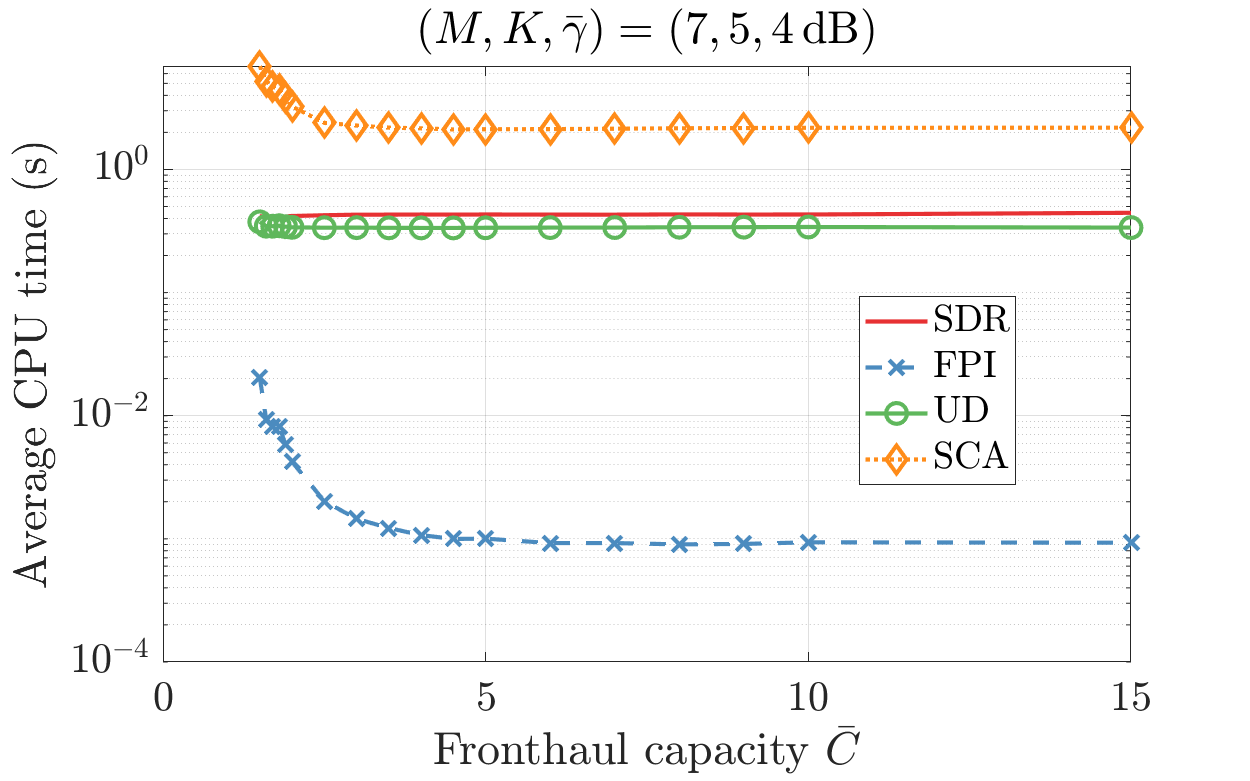}
	\label{fig:754_T}}\\
	\subfloat[]{\includegraphics[width=0.4\textwidth]{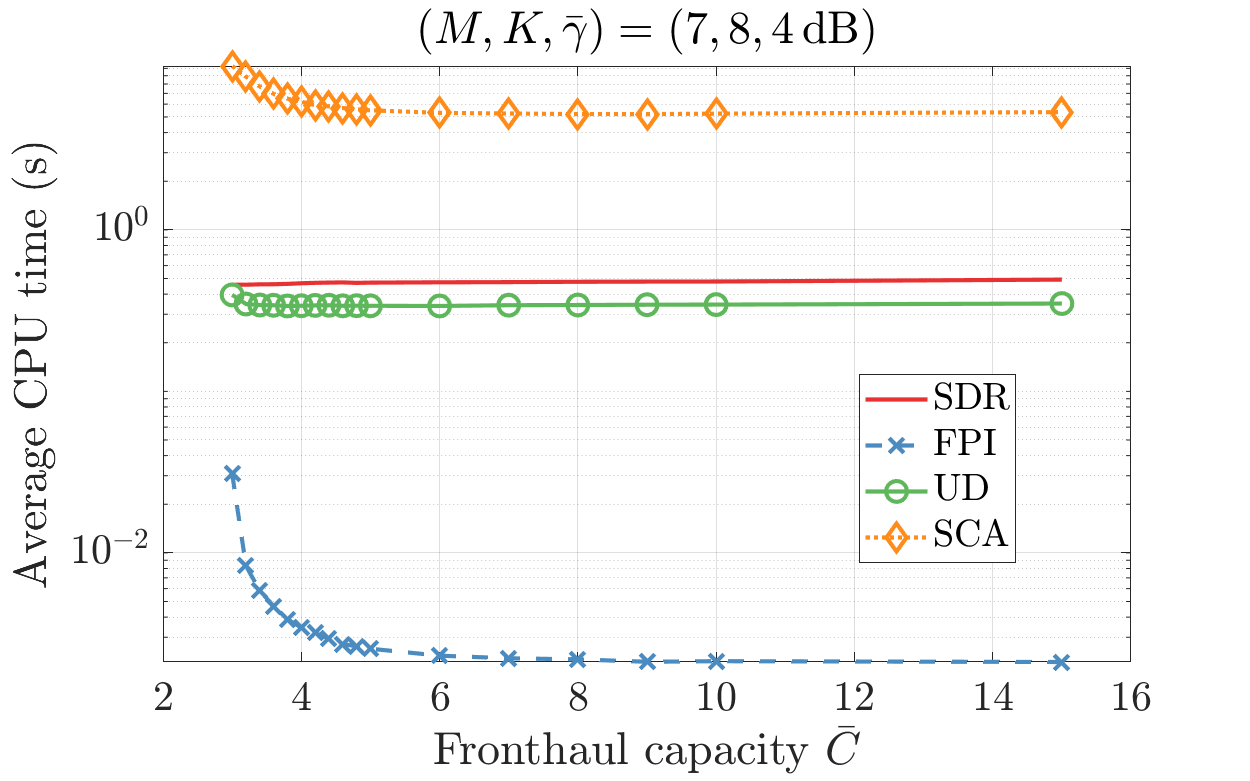}
	\label{fig:784_T}}\\
	\caption{The average CPU time of different algorithms with different system parameters $(M, K, \bar \gamma, \bar C)$. }
	\label{fig:SINR_T}
\end{figure}

The numerical results in Figs.~\ref{fig:SINR}~and~\ref{fig:SINR_T} also show the importance of exploiting the problem structure in improving the solution efficiency. 
We make two remarks on this aspect. 
First, by exploiting the problem structure, we have shown the tightness of the SDR in \eqref{SDR}, which enables us to solve the original seemingly nonconvex problem~\eqref{SDR} by solving a single convex SDR in \eqref{SDR}. 
In sharp contrast, the SCA algorithm needs to solve a series of convex approximation problems to solve the original problem, which makes it less efficient than the SDR algorithm. 
Second, one key difference between the UD and FPI algorithms is that the proposed FPI algorithm uses the fixed point iteration in \eqref{equ:pri_iter} to solve the primal problem while UD needs to call the solver to solve a reduced primal downlink problem (after solving the dual problem). 
The proposed FPI algorithm judiciously utilizes the relations between the primal and dual variables (i.e., Eqs. \eqref{equ:slack_1} and \eqref{equ:pri_var1}) and hence significantly improves the computational efficiency of solving the primal problem when compared with the UD algorithm (which does not leverage the special structure of the primal problem). 

\begin{figure}[t]
	\centering
	\subfloat[]{\includegraphics[width=0.4\textwidth]{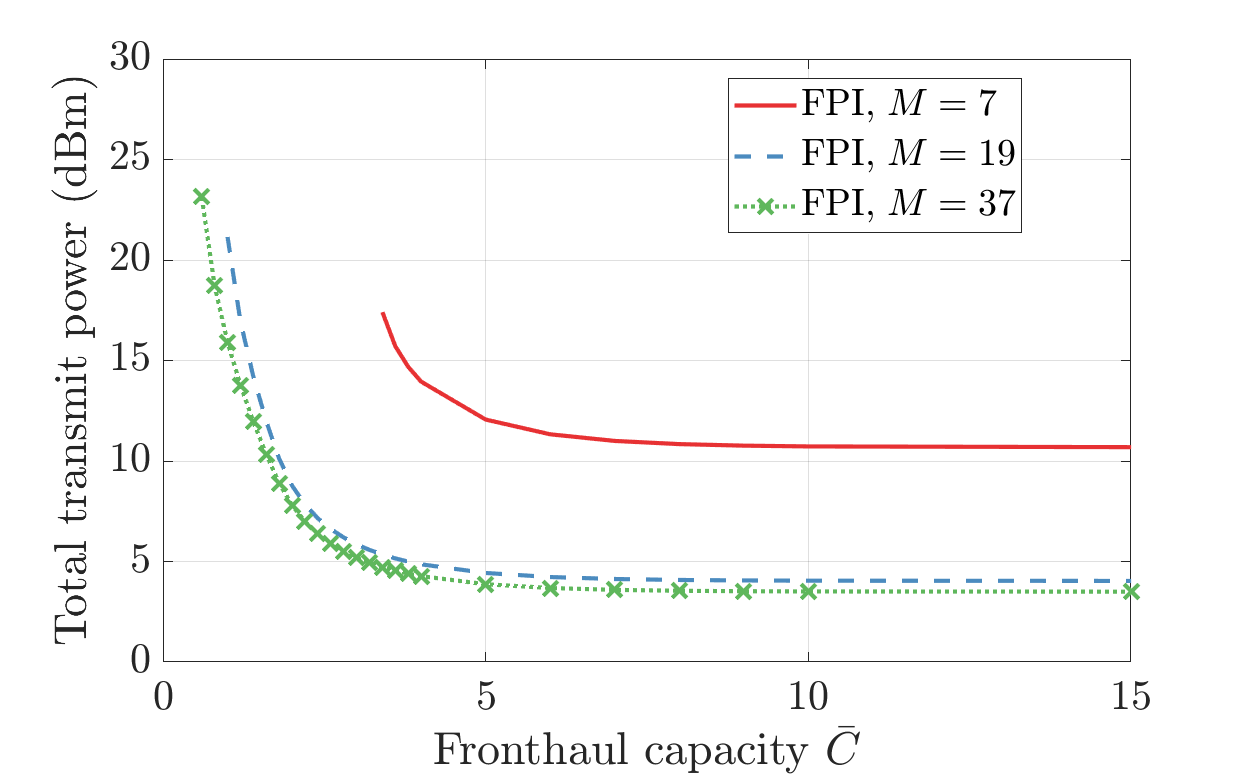}
	\label{fig:C_fronthaul}}\\
	\subfloat[]{\includegraphics[width=0.4\textwidth]{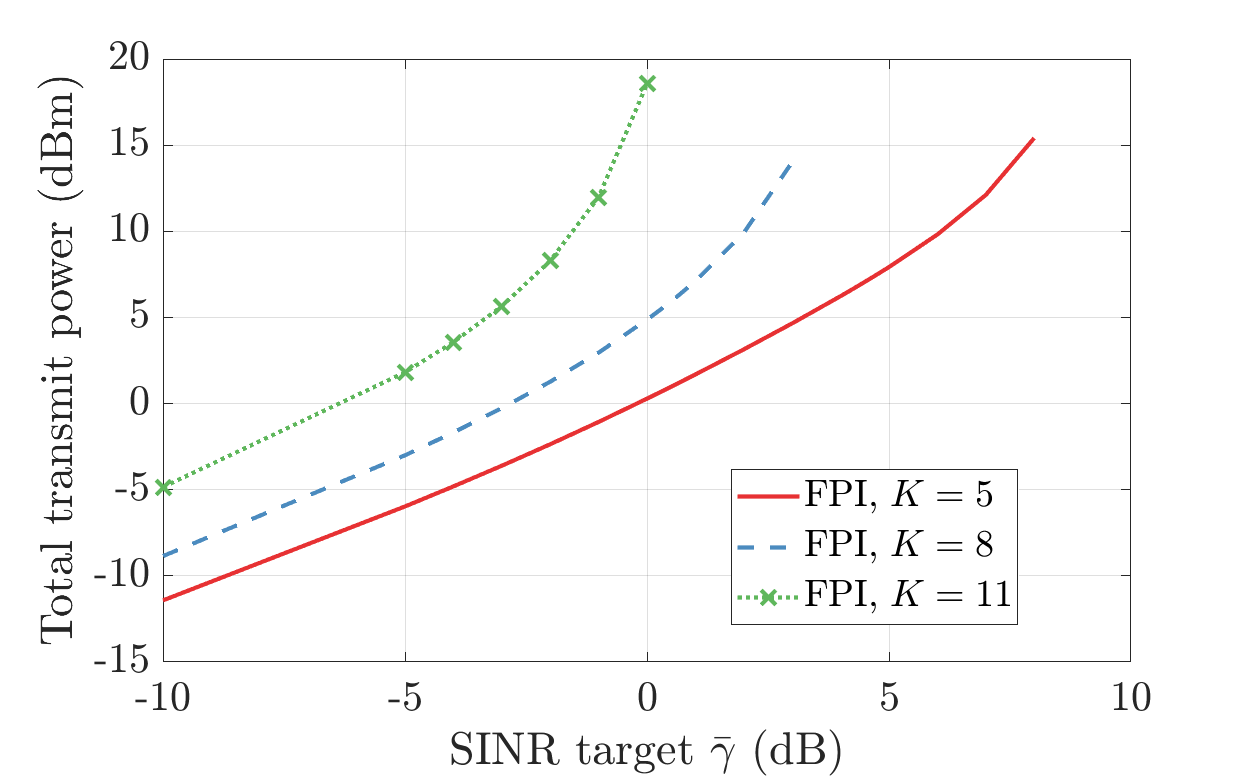}
	\label{fig:C_SINR}}\\
	\caption{The average objective value (i.e., the total transmit power) with different system parameters $(M, K, \bar \gamma, \bar C)$. }
	\label{fig:para}
\end{figure}

\subsection{Impact of System Parameters}

In this subsection, we investigate the impact of the system parameters $(\bar C, M, \bar \gamma, K)$ on the objective value (i.e., the total transmit power) of problem (P). 

Fig. 7 plots the total transmit power with different choices of the system parameters $(\bar C, M, \bar \gamma, K)$. 
In Fig.~\ref{fig:para}~\subref{fig:C_fronthaul}, as $\bar C$ or $M$ increases, the total transmit power decreases. 
This is because, a larger $\bar C$ allows for a smaller compression noise, and the smaller compression noise further causes less interference to users, thereby requiring less total transmission power to satisfy all users' SINR requirements; a larger $M$ gives more space freedom in the beamformers, which reduces the total transmit power. In summary, with a larger $\bar C$ and/or $M,$ the feasible region of the problem becomes larger and thus the total transmit power becomes smaller. 
However, when the fronthaul capacity and the number of relays are greater than a certain value, i.e., $(\bar C, M) = (10, 19)$, further increasing them only brings a marginal improvement to the corresponding system's performance.

In Fig.~\ref{fig:para}~\subref{fig:C_SINR}, as the SINR target $\bar \gamma$ increases, the total transmit power increases. 
This is because with larger a higher SINR target $\bar \gamma$, the feasible region of the corresponding problem becomes smaller, thereby resulting in a larger total transmit power. 
Moreover, for a given SINR target $\bar \gamma$, the total transmit power increases as the increasing of $K$. 
This is because, a larger $K$ generally induces larger multiuser interferences and more power is needed to manage these interferences in order to achieve all users' SINR targets.

The observations from Fig.~\ref{fig:para} provide the following engineering insights for the system design. 
First, increasing the fronthaul capacity $\bar{C}$ and/or the number of relays $M$ can significantly reduce the total transmit power up to a certain point, after which further increases yield negligible improvements.
When designing a network, it is critical to allocate fronthaul capacity and relays up to the point where the system performance saturates. 
This avoids unnecessary costs associated with further increases in $\bar{C}$ and $M$. 
Second, there is a clear trade-off between the SINR target $\bar \gamma$ and the number of users $K$. 
It is impossible to serve too many users simultaneously with stringent QoS requirements. 
In practice, scheduling users based on their channel conditions and available system resources (e.g., fronthaul capacity, the number of relays, and the power budget at the BS) is an important way to guarantee users' QoS requirements. 

\section{Conclusion}

In this paper, we consider the QoS-based joint beamforming and compression design problem in the cooperative cellular network. 
A major design challenge is to find the global beamforming and compression strategy to minimize the total network transmit power. 
We first show in this paper, that the seemingly nonconvex design problem admits a convex SDP reformulation by proving that its SDR is tight. 
Based on the above result, we further propose an efficient algorithm for globally solving the considered problem. 
The basic idea of the proposed algorithm is to solve an enhanced KKT conditions of the SDR of the considered problem via two fixed point iterations. 
Two key features of the proposed algorithm are: 
(1) it is guaranteed to find the global solution of the problem when the problem is feasible and is able to detect its infeasibility when the problem is not feasible; 
(2) it is highly efficient because both of fixed point iterations in the proposed algorithm are linearly convergent and each evaluation of the functions in the fixed point iterations are computationally cheap. 
Numerical results show that the proposed algorithm significantly outperforms the SOTA benchmarks in terms of the computational efficiency.
While this paper focuses on theoretical insights and algorithmic efficiency, addressing practical considerations such as robustness to modeling inaccuracies and imperfect channel state information are crucial for enhancing the real-world applicability of the proposed solutions. 
Hence, we shall consider these important practical factors in future works.

\appendices

\section{Proof of Eqs. \eqref{equ:dual_var2} and \eqref{equ:pri_const1} } \label{apd:kkt}

We first prove that Eq. \eqref{equ:dual_var2} holds at the optimal solution of the dual problem \eqref{dual}. 
Note that $\m{\Lambda}_m^{(1:m-1, 1:M)}$ and $\m{\Lambda}_m^{(m:M, 1:m-1)}$ do not affect the objective value or any constraints in problem \eqref{dual}. 
Hence, one can choose $\Lm$ such that $ \m \Lambda_m^{(1:m-1, 1:M)}$ and $ \m \Lambda_m^{(m:M,1:m-1)}$ are all zero. 
Combining the complementary slackness of $\m B_m$ and $\m{\Lambda}_m$
with the fact that $\m B_m$ is of rank $M-m+1$ or $M-m$ yields $\rk (\m{\Lambda}_m) \leq 1$. 

Now we prove that Eq. \eqref{equ:pri_const1} is true. 
To show this, it suffices to show that the optimal dual solution $\beta_k > 0$ for all $k \in \cK$ (due to the complementary slackness). 
Next, we use the following contradiction argument to show that the optimal dual solution $\beta_k > 0$ for all $k \in \cK$. 
Suppose that $(\v \beta, \Lm)$ is an optimal dual solution, but there exists some $k_0$ such that $\beta_{k_0} = 0$. 
Then let $\tilde{\beta}_{k_0} = \bar \gamma_{k_0}\left( \v{h}_{k_0}^\hermitian \m C_k\left( \v \beta, \left\{ \m \Lambda_m \right\} \right)^{-1} \v{h}_{k_0} \right)^{-1}$ and $\tilde{\beta}_k = \beta_k \text{~for~all~} k \neq k_0$. 
Define $\tilde{\v \beta} = [\tilde \beta_1, \tilde \beta_2, \d, \tilde \beta_K]^\transpose$. 
Then $(\tilde{\v \beta}, \{\m{\Lambda}_m \})$ is a feasible dual solution with a larger objective value, 
which contradicts the optimality of the solution $(\v \beta, \{\m{\Lambda}_m \})$.

\section{Proof of Theorem \ref{thm:alg_prop}}
\label{apd:alg_prop}
In this part, we shall first show the convergence of the dual fixed point iteration \eqref{equ:dual_iter} and the primal fixed point iteration \eqref{equ:pri_iter} by showing that both mappings $I(\cdot)$ in \eqref{equ:dual_fix_point} and $J(\cdot)$ in \eqref{equ:primal_fix_point} are standard interference (SI) mappings \cite{yates1995FrameworkUplinkPower}. 
Then we shall show the linear convergence rate of the dual fixed point iteration \eqref{equ:dual_iter} and the primal fixed point iteration \eqref{equ:pri_iter}. 

\subsection{Convergence of Fixed Point Iterations \eqref{equ:dual_iter} and \eqref{equ:pri_iter}}

The existence of the fixed points of $I(\cdot)$ in \eqref{equ:dual_fix_point} and $J(\cdot)$ in \eqref{equ:primal_fix_point} is guaranteed by the strict feasibility of problem \eqref{SDR}. 
Combining this with \cite[Theorem 2]{yates1995FrameworkUplinkPower}, it suffices to show that both $I(\cdot)$ in \eqref{equ:dual_fix_point} and $J(\cdot)$ in \eqref{equ:primal_fix_point} are SI mappings \cite{yates1995FrameworkUplinkPower} in order to show the convergence of fixed point iterations \eqref{equ:dual_iter} and \eqref{equ:pri_iter}. 
A mapping $f: \mathbb{R}_+^n \rightarrow \mathbb{R}^n$ is said to be an SI mapping if it satisfies the following three properties\footnote{The positivity property introduced in \cite{yates1995FrameworkUplinkPower} has been replaced with the nonnegativity property, as shown in \cite[Lemma 1]{leung2004ConvergenceTheoremGeneral} and \cite[Fact 1]{cavalcante2019ConnectionsSpectralProperties}.}.

\textbf{Nonnegativity}: For any $\v x \in \mathbb{R}^n_{+}$, $f(\v x) \geq \v 0$. 

\textbf{Strict subhomogeneity}: For any $\alpha>1$ and $\v x \in \mathbb{R}^n_{+} \backslash \left\{ \v 0 \right\}$, $f(\alpha \v x) < \alpha f(\v x)$. 

\textbf{Monotonicity}: For any $\v x_1, \v x_2 \in \mathbb{R}^n_{+}$ with $\v x_1 \geq \v x_2$, $f(\v x_1) \geq f(\v x_2)$. 

Notice that a mapping $f = [f_1, f_2, \d, f_n]^\transpose$ is an SI mapping if and only if all of its components $f_i: \mathbb R_+ \rightarrow \mathbb R$ satisfy the above three properties.

\subsubsection{$I(\cdot)$ in \eqref{equ:dual_fix_point} is an SI mapping}
We show that $I(\cdot)$ in \eqref{equ:dual_fix_point} is an SI mapping by showing that $I_k(\cdot)$ in \eqref{equ:raw_dual_iter} for all $k \in \cK$ satisfy the nonnegativity, the strict subhomogeneity, and the monotonicity one by one. 
From \eqref{equ:raw_dual_iter}, for any $ k \in \cK $, we have 
\begin{equation}
	I_k(\v \beta) = \frac{\bar \gamma_k}{\v{h}_k^\hermitian \m C_k(\v \beta)^{-1} \v{h}_k}, 
	 \label{equ:def_Ik}
\end{equation} 
where we use $\m C_k(\v \beta)$ to denote 
\begin{equation}
	\m C_k(\v \beta, \{\m \Lambda_m(\v \beta)\}) = \m I + \sum_{m\in \cM} \m \Lambda_m^{(m,m)}(\v \beta) \m E_m + \sum_{j \neq k} \beta_j \v h_j \v h_j^\hermitian
	\label{equ:def_Ck}
\end{equation}
to simplify the notations. 
In the following, we shall show the desired properties of $I_k(\cdot)$ based on the properties of $\m \Lambda_m^{(m,m)}(\cdot)$ in Lemma \ref{lem:lmmm} given in Appendix~\ref{ss:useful_properties_Lambda} of the Supplementary Material. 

\textbf{Proof of the nonnegativity of $I_k(\cdot)$}: For any $m \in \cM$, $\m \Lambda_m^{(m,m)}(\cdot)$ is positive by Lemma \ref{lem:lmmm}, which implies that $\m C_k(\v \beta)$ is positive definite for any $\v \beta \in \mathbb{R}_+^K$. 
Hence, by \eqref{equ:def_Ik}, $ I_k(\cdot) $ is nonnegative. 

\textbf{Proof of the strict subhomogeneity of $I_k(\cdot)$}: For any $ \alpha > 1 $ and $\v\beta \in \mathbb{R}^K_{+}\backslash \left\{ \v 0 \right\}$, it follows from the strict subhomogeneity of $ \m{\Lambda}_m^{(m,m)}(\cdot) $ in Lemma \ref{lem:lmmm} and \eqref{equ:def_Ck} that 
\begin{equation}
	\label{equ:Lk_homo}
	\begin{aligned}
		\m C_k(\alpha\v \beta) - \m I 
		&= \sum_{m\in\cM} \m{\Lambda}_m^{(m,m)}\left(\alpha\v \beta\right)\m{E}_m+ \sum_{j\neq k} \alpha\beta_j \v h_j \v h_j^\hermitian \\
		&\prec \alpha\left(\sum_{m\in\cM} \m{\Lambda}_m^{(m,m)}\left(\v \beta\right)\m{E}_m+\sum_{j\neq k} \beta_j \v h_j \v h_j^\hermitian\right)\\
		&= \alpha (\m C_k(\v \beta) - \m I). 
	\end{aligned}
\end{equation}
As such, by \eqref{equ:def_Ik}, we get
\begin{equation*}
	\begin{aligned}
		I_k(\v{\alpha \beta}) &= \bar \gamma_k \left( \v{h}_k^\hermitian \left( 
			\m I + \m C_k(\alpha \v \beta) - \m I
			\right)^{-1} \v{h}_k \right)^{-1}  \\
			&< \bar \gamma_k\left( \v{h}_k^\hermitian \left( 
				\alpha \m I + \alpha \left(\m C_k(\v \beta) - \m I\right)
				\right)^{-1} \v{h}_k \right)^{-1} \\
			&= \alpha I_k(\v \beta). 
	\end{aligned}
\end{equation*}
\textbf{Proof of the monotonicity of $I_k(\cdot)$}: For any $\v \beta_1 , \v \beta_2 \in \mathbb{R}^K_{+}$ with $\v \beta_\ell = [\beta_{\ell,1}, \beta_{\ell,2}, \d, \beta_{\ell,K}]^\transpose$ for $\ell=1,2$ and $\v \beta_1 \geq \v \beta_2$, the monotonicity of $ \m{\Lambda}_m^{(m,m)}\left(\cdot\right) $ in Lemma \ref{lem:lmmm} and \eqref{equ:def_Ck} gives 
\begin{equation*}
	\begin{aligned}
		\m C_k(\v \beta_2)
		&= \m I + \sum_{m\in\cM} \m{\Lambda}_m^{(m,m)}\left(\v \beta_2\right)\m{E}_m+ \sum_{j\neq k} \beta_{2, j} \v h_j \v h_j^\hermitian \\
		&\preceq \m I + \sum_{m\in\cM} \m{\Lambda}_m^{(m,m)}\left(\v \beta_1\right)\m{E}_m+ \sum_{j\neq k} \beta_{1, j} \v h_j \v h_j^\hermitian\\
		&= \m C_k(\v \beta_1). 
	\end{aligned}
\end{equation*}
Combining this with \eqref{equ:def_Ik} yields 
\begin{equation*}
	\begin{aligned}
		I_k(\v \beta_2) &= \bar \gamma_k\left( \v{h}_k^\hermitian \m C_k(\v \beta_2)^{-1} \v{h}_k \right)^{-1} \\
			&\leq \bar \gamma_k\left( \v{h}_k^\hermitian \m C_k(\v \beta_1)^{-1} \v{h}_k \right)^{-1} \\
			&= I_k(\v \beta_1). 
	\end{aligned}
\end{equation*}

\subsubsection{$J(\cdot)$ in \eqref{equ:primal_fix_point} is an SI mapping}
We show the nonnegativity, the strict subhomogeneity, and the monotonicity of $J_k(\cdot)$ one by one, where $J_k(\cdot)$ in \eqref{equ:raw_primal_iter} can be explicitly written as
\begin{equation}
	J_k(\v p) = \frac{\bar \gamma_k \left( \sum_{j\neq k} p_j \v  h_k^\hermitian \hat{\m V}_j \v h_k + \v h_k^\hermitian \m Q(\v p) \v h_k + \sigma_k^2 \right)}{\v  h_k^\hermitian \hat{\m V}_k \v h_k}. 
	\label{equ:def_Jk}
\end{equation}
In the following, we shall show the desired properties of $J_k(\cdot)$ based on the properties of $\m Q(\cdot)$ in Lemma \ref{lem:Q} given in Appendix~\ref{ss:useful_properties_Q} of the Supplementary Material. 

\textbf{Proof of the nonnegativity of $J_k(\cdot)$}: For any $\v p \in \mathbb{R}^K_{+}$, the nonnegativity of $ \m Q(\cdot) $ in Lemma \ref{lem:Q} and \eqref{equ:def_Jk} yields 
$$
\begin{aligned}
	J_k(\v p) \geq \frac{\bar \gamma_k \sigma_k^2}{\v  h_k^\hermitian \hat{\m V}_k \v h_k} \geq 0. 
\end{aligned}
$$

\textbf{Proof of the strict subhomogeneity of $J_k(\cdot)$}: For any $\alpha > 1$ and $\v p \in \mathbb{R}^K_{+} \backslash \left\{ \v 0 \right\}$, the linearity of $ \m Q(\cdot) $ in Lemma \ref{lem:Q} together with \eqref{equ:def_Jk} gives
	\begin{equation*}
		\begin{aligned}
			J_k(\alpha \v p) &= \frac{\bar \gamma_k \left( \alpha \sum_{j\neq k} p_j \v  h_k^\hermitian \hat{\m V}_j \v h_k + \alpha \v h_k^\hermitian \m Q(\v p) \v h_k + \sigma_k^2 \right)}{\v  h_k^\hermitian \hat{\m V}_k \v h_k} \\ 
			&< \frac{\bar \gamma_k \left( \alpha \sum_{j\neq k} p_j \v  h_k^\hermitian \hat{\m V}_j \v h_k + \alpha  \v h_k^\hermitian \m Q(\v p) \v h_k + \alpha \sigma_k^2 \right)}{\v  h_k^\hermitian \hat{\m V}_k \v h_k}\\
			&=  \alpha J_k(\v p). 
		\end{aligned}
	\end{equation*}

\textbf{Proof of the monotonicity of $J_k(\cdot)$}: For any $\v p_1, \v p_2 \in \mathbb{R}_+^K$ with $\v p_1 \geq \v p_2$, combining the monotonicity of $ \m Q(\cdot) $ in Lemma \ref{lem:Q} and \eqref{equ:def_Jk} shows that  
\begin{equation*}
	\begin{aligned}
		J_k(\v p_2) &= \frac{\bar \gamma_k \left( \sum_{j\neq k} p_{1,j} \v  h_k^\hermitian \hat{\m V}_j \v h_k + \v h_k^\hermitian \m Q(\v p_2) \v h_k + \sigma_k^2 \right)}{\v  h_k^\hermitian \hat{\m V}_k \v h_k}  \\ 
		&\leq \frac{\bar \gamma_k \left( \sum_{j\neq k} p_{2,j} \v  h_k^\hermitian \hat{\m V}_j \v h_k + \v h_k^\hermitian \m Q(\v p_1) \v h_k + \sigma_k^2 \right)}{\v  h_k^\hermitian \hat{\m V}_k \v h_k} \\
		&= J_k(\v p_1). 
	\end{aligned}
\end{equation*}

\subsection{Linear Convergence Rates of Fixed Point Iterations \eqref{equ:dual_iter} and \eqref{equ:pri_iter}}
In this part, we show the linear convergence rate of both iterations. 
First, combining the linearity of $\m Q(\cdot)$ in Lemma \ref{lem:Q} and \eqref{equ:def_Jk} gives that $J(\cdot)$ is an affine function. 
Hence, the convergence of the primal fixed point iteration \eqref{equ:pri_iter} immediately implies its linear convergence rate (and the rate depends on the spectral radius of matrix $\m G$ whose entries are given in \eqref{equ:entry_G}). 
Next, we shall focus on showing the linear convergence rate of the dual fixed point iteration \eqref{equ:dual_iter}. 
To this end, we first establish the following two lemmas, which are later combined together to bound the ratio $\frac{\mu(\v \beta^{(i+1)}, \v \beta^\star)}{\mu(\v \beta^{(i)}, \v \beta^\star)}$ in the left-hand-side of \eqref{equ:c_r_limit}. 
Then, we take the limit superior of the obtained bound to show the desired result in \eqref{equ:c_r_limit}. 
Recall the metric $\mu (\cdot, \cdot)$ defined in \eqref{equ:metric}. 
\begin{lemma}
	\label{lem:cr1}
	For any $ \v \beta, \tilde{\v \beta} \in \mathbb{R}_{++}^K $, we have 
	\begin{equation}
		\frac{\mu (I(\v \beta), I(\tilde{\v \beta}))}{\mu (\v \beta, \tilde{\v \beta})} \leq \max_{k} \Bigg\{ \log_{\alpha} \Big( \frac{I_k(\alpha \tilde{\v \beta})}{I_k(\tilde{\v \beta})} \Big), \log_\alpha \Big( \frac{I_k(\alpha\v \beta)}{I_k(\v \beta)} \Big) \Bigg\}, 
		\label{equ:desired_result1}
	\end{equation}
	where $ \alpha = \mathrm{e}^{\mu(\v \beta, \tilde{\v \beta})} $. 
\end{lemma}
\begin{proof}
	By the definition of $ \mu(\cdot, \cdot) $ in \eqref{equ:metric}, we get $\v \beta \leq \alpha \tilde{\v \beta}$ and $ \tilde{\v \beta} \leq \alpha \v \beta $. 
	Combining this with the monotonicity of $I_k(\cdot)$ gives 
	\begin{equation*}
		I_k\left(\v \beta\right)\leq I_k(\alpha \tilde{\v \beta}) \text{~and~} I_k(\tilde{\v \beta})\leq I_k(\alpha \v \beta). 
	\end{equation*}
	As a result, \begin{equation*}
		\begin{aligned}
			\mu(I(\v \beta), I(\tilde{\v \beta})) &= \max_{k} \left|\log_{\mathrm e} \Big( \frac{I_k(\v \beta)}{I_k(\tilde{\v \beta})} \Big)\right|\\
			&\leq \max_{k} \left\{  
			\log_{\mathrm e} \Big( \frac{I_k(\alpha\tilde{\v \beta})}{I_k(\tilde{\v \beta})} \Big), 
			\log_{\mathrm e} \Big( \frac{I_k(\alpha\v \beta)}{I_k(\v \beta)} \Big) \right\}.  \\
		\end{aligned}
	\end{equation*}
	Dividing both sides of the above inequality by $ \mu(\v \beta, \tilde{\v \beta})$ yields the desired result \eqref{equ:desired_result1}. 
\end{proof}
\begin{lemma}
	\label{lem:cr2}
	For any $ \alpha > 1$ and $ \v \beta \in \mathbb{R}_{++}^K $, we have 
	\begin{equation}
		\log_{\alpha} \left( \frac{I_k(\alpha\v \beta)}{I_k(\v \beta)} \right) < \log_{\alpha}\left( \frac{1+\alpha \lambda^k_1(\v \beta)}{1 + \lambda^k_1(\v \beta)} \right), 
		\label{equ:desired_conclusion2}
	\end{equation}
	where we use $ \lambda_1^k(\v \beta) $ to denote $ \|\m C_k(\v \beta) - \m I\|_2 $. 
\end{lemma}
\begin{proof}
	For any $ \alpha > 1$ and $ \v \beta \in \mathbb{R}_{++}^K $, it follows from \eqref{equ:def_Ik} and \eqref{equ:Lk_homo} that 
	\begin{equation}
		\begin{aligned}
			\frac{I_k\left(\alpha \v \beta\right)}{I_k(\v \beta)} 
			&= \frac{\v{h}_k^\hermitian \m C_k(\v \beta)^{-1} \v{h}_k}{\v{h}_k^\hermitian \m C_k( \alpha \v \beta)^{-1} \v{h}_k} \\
			&< \frac{\v{h}_k^\hermitian \m C_k(\v \beta)^{-1} \v{h}_k}{\v{h}_k^\hermitian (\m I + \alpha (\m C_k(\v \beta) - \m I))^{-1} \v{h}_k}. 
		\end{aligned}
		\label{equ:Ikratio1}
	\end{equation}
	In the rest part of the proof, we drop the dependence of all variables on $\v \beta$ and $k$ for notational simplicity. 
	Suppose $\m U \m \Lambda \m U^\hermitian$ is the spectral decomposition of $\m C - \m I$, where $\m \Lambda = \diag(\lambda_1, \lambda_2, \d, \lambda_M)$ with decreasing $\lambda_m$. 
	Let $\v v = \m U^\hermitian \v h = [v_1, v_2, \d, v_M]^\transpose$. 
	Then, we get  
	\begin{equation}
		\begin{aligned}
			\frac{\v h^\hermitian \m C^{-1} \v h}{\v h^\hermitian (\m I + \alpha(\m C - \m I))^{-1} \v h} &= \frac{\v v^\hermitian(\m I + \m \Lambda)^{-1} \v v}{\v v^\hermitian(\m I + \alpha\m \Lambda)^{-1} \v v} \\
			&= \frac{\sum_{m\in\cM} \frac{1}{1+\lambda_m} |v_m|^2}{\sum_{m\in\cM} \frac{1}{1+\alpha\lambda_m} |v_m|^2} \\
			&\leq \frac{\sum_{m\in\cM} \frac{1}{1+\lambda_1} |v_m|^2}{\sum_{m\in\cM} \frac{1}{1+\alpha\lambda_1} |v_m|^2} \\
			&= \frac{1+\alpha \lambda_1}{1+\lambda_1}. 
		\end{aligned}
		\label{equ:Ikratio2}
	\end{equation}
	Finally, combining \eqref{equ:Ikratio1} and \eqref{equ:Ikratio2} and taking the $ \alpha $-logarithm from both sides yield the desired result \eqref{equ:desired_conclusion2}. 
\end{proof}

Define $ \kappa\left(\alpha, \lambda\right)= \log_\alpha \left(\frac{1+\alpha \lambda}{1+\lambda}\right) $. 
Then it is simple to verify that the function has the following properties: 
\begin{enumerate}
	\item [(i)] $ \kappa(\alpha, \lambda) $ is an increasing function of both $ \alpha $ and $ \lambda $; 
	\item [(ii)] $ \kappa\left(\alpha, \lambda\right)\in (0,1) $ for any  $ \alpha >1$ and $ \lambda>0 $; and 
	\item [(iii)] $ \lim_{\alpha \rightarrow \infty} \kappa\left(\alpha, \lambda\right) = \frac{\lambda}{1 + \lambda} $ for any $ \lambda > 0 $. 
\end{enumerate}

Now, we combine Lemmas \ref{lem:cr1} and \ref{lem:cr2} to show the linear convergence rate of the dual fixed point iteration \eqref{equ:dual_iter}. 
First, substituting $ \v \beta $ and $\tilde{\v \beta}$ with $ \v \beta^{(i)} $ and $\v \beta^*$ in Lemma \ref{lem:cr1} and using Lemma \ref{lem:cr2} gives 
\begin{equation*}
	\frac{\mu(\v \beta^{(i+1)}, \v \beta^*)}{\mu(\v \beta^{(i)}, \v \beta^*)} 
	< \max_k \left\{ \kappa(\alpha_i, \lambda_1^k(\v \beta^{(i)})), \kappa(\alpha_i, \lambda_1^k(\v \beta^{*})) \right\}, 
\end{equation*}
where $ \alpha_i = \mathrm{e}^{\mu(\v \beta^{(i)}, \v \beta^*)} $. 
Since $ \kappa(\alpha, \cdot) $ is an increasing function, it follows from the fact $\lambda(\v \beta) = \max_k \left\{ \lambda_1^k(\v \beta) \right\}$ that 
\begin{equation}
	\frac{\mu(\v \beta^{(i+1)}, \v \beta^*)}{\mu(\v \beta^{(i)}, \v \beta^*)} 
	< \max \left\{ \kappa(\alpha_i, \lambda(\v \beta^{(i)})), \kappa(\alpha_i, \lambda(\v \beta^{*})) \right\}. 
	\label{equ:cr_key}
\end{equation}
Taking the limit superior on both sides of \eqref{equ:cr_key} and using the properties of $\kappa (\cdot, \cdot)$, we obtain 
\begin{equation*}
	\begin{aligned}
		&\limsup_{i\rightarrow \infty}\frac{\mu(\v \beta^{(i+1)}, \v \beta^*)}{\mu(\v \beta^{(i)}, \v \beta^*)} \\
		&\quad\leq \max\left\{\lim_{i \rightarrow \infty} \kappa\left(\alpha_i, \lambda(\v \beta^{(i)})\right), \lim_{i \rightarrow \infty} \kappa\left(\alpha_i, \lambda(\v \beta^{*})\right)\right\}   \\
		&\quad= \frac{\lambda(\v \beta^*)}{1 + \lambda(\v \beta^*)}.
	\end{aligned}
\end{equation*}
This shows the (asymptotic) linear convergence rate of the dual fixed point iteration \eqref{equ:dual_iter} given in \eqref{equ:c_r_limit}.

\section*{Acknowledgments}
	The authors wish to thank Dr.~R.~L.~G.~Cavalcante of Fraunhofer Heinrich-Hertz-Institut for his insightful comments, which have significantly improved the presentation of the paper.

\bibliographystyle{IEEEtran}
\bibliography{tsp_duality}

\clearpage
\begin{center}
  \Large\textbf{Supplementary Material}
\end{center}
\appendices
\setcounter{section}{2}
\setcounter{page}{1}

\section{Equivalence between Problem \eqref{equ:original obp} and Problem (P)}
\label{apd:equivalence}
By \cite[Proposition 4]{liu2021UplinkdownlinkDualityMultipleaccess}, the SINR constraints in problem \eqref{equ:original obp} are equivalent to those in problem (P), and the fronthaul rate constraints in problem \eqref{equ:original obp} are equivalent to those in problem (P) when $\m Q$ is \emph{positive definite}. 
In the following, we show the equivalence between the $m$-th fronthaul rate constraint in problem \eqref{equ:original obp} and that in problem (P) for all $m\in\cM$ when $\m Q$ is \emph{positive semidefinite}. 

First, if $(\{\v v_k\}, \m Q)$ satisfies the $m$-th fronthaul rate constraint in problem (P), then, by the Schur complement, we have 
\begin{equation}
\label{equ:def_C}
    2^{\bar C_m} q_m - \sum_{k\in\cK} |v_{k,m}|^2 - \m Q^{(m,m)} \geq 0.
\end{equation}
This implies $C_m \leq \bar C_m$ for both cases $q_m = 0$ and $q_m > 0$. 
Next, we show the other direction. 
Suppose that $(\{\v v_k\}, \m Q)$ satisfies the $m$-th fronthaul rate constraint in problem \eqref{equ:original obp}. 
We consider two cases where $q_m=0$ and $q_m>0$ separately.
\begin{itemize}
    \item Case $q_m = 0$. In this case, the definition of $C_m$ in \eqref{equ:def_Cm} implies that $v_{k,m} = 0$ for all $k\in\cK$ and $\v Q^{(1:M,m)} = (\v Q^{(m,1:M)})^\hermitian = \v 0$, and then the $m$-th fronthaul rate constraint in problem (P) becomes
$$
2^{\bar C_m} \begin{bmatrix}
				\m 0 & \m 0  \\
				\m 0 & \m Q^{(m:M, m:M)}
			\end{bmatrix} \succeq \m 0, 
$$
which holds naturally. 
    \item Case $q_m > 0$. In this case, from the definition of $C_m$ in \eqref{equ:def_Cm}, one has \eqref{equ:def_C}. 
Notice that $\m Q^{(m+1:M,m)}$ always lies in the column space of $\m Q^{(m+1:M,m+1:M)}$ since $\m Q\succeq \m 0$ \cite[Theorem 1.20]{zhang2005SchurComplementIts}. 
Combining this with \eqref{equ:def_C} and $\m Q^{(m+1:M,m+1:M)} \succeq \m 0$ yields
\begin{equation*}
    \small
    2^{\bar C_m} \begin{bmatrix}
				\m 0 & \m 0  \\
				\m 0 & \m Q^{(m:M, m:M)}
			\end{bmatrix} - \left(\sum_{k\in\cK} |v_{k,m}|^2+ \m Q^{(m,m)}\right) \m{E}_m \succeq \m 0, 
\end{equation*}
i.e., the $m$-th fronthaul rate constraint in problem (P) holds \cite[Theorem 1.20]{zhang2005SchurComplementIts}. 
\end{itemize}

\section{The Positive Definiteness of the Optimal Solution~\texorpdfstring{$\m Q$}{Q} to Problem \eqref{SDR}}
\label{apd:PD}
In this part, we show that the optimal solution $\m Q$ to problem \eqref{SDR} is positive definite with probability one under a mild assumption on the random channels. 
This is done by examining the case where there exists a singular optimal solution $\m Q$ to problem~\eqref{SDR} and showing that this will lead to a linear equation that has probability zero of occurring. 
The proof outline is as follows. 
We start by introducing a reduced problem which is obtained by removing the unused relay $m$ from problem~\eqref{SDR} and hence the reduced problem has the same optimal value as problem~\eqref{SDR}. 
Then we show that this implies that the channels from relay $m$ to the users must satisfy a linear equation (see \eqref{equ:linear_equ} in the proof of Proposition~\ref{thm:PD} further ahead), whose coefficients solely depend on the channels from the other relays to users. 
However, the linear equation has probability zero of occurring as the coefficients of the linear equation is statistically independent of the randomly generated channel coefficients from relay $m$ to all users. 
Therefore, the optimal solution $\m Q$ to problem \eqref{SDR} is positive definite with probability one. 

First, we construct the reduced problem in the case of a singular optimal $\m Q$. 
If an optimal solution $ \m Q $ to problem~\eqref{SDR} is singular, i.e., $q_m$ in \eqref{equ:def_Cm} is zero for some $m \in \cM$, then relay $ m $ does not play any role in the whole transmission process; please refer to the discussion under Eq. \eqref{equ:def_Cm}. 
Hence, if we remove relay $m$ from problem~\eqref{SDR}, the optimal value of the reduced problem~\eqref{SDR} will be the same as that of the original problem~\eqref{SDR}. 
With proper reordering and relabeling, let relay $1$ be the removed relay. 
Without loss of generality, we assume that relay $1$ is the last to be compressed for the ease of presentation. (The proof also applies to the general case where the removed relay is compressed in any other order.)
Then the corresponding reduced problem is given by
\begin{equation}
	\begin{aligned}
		\min_{\{\tilde{\m V}_k\}, \tilde{\m Q} \succeq \m 0 } &\quad \sum_{k\in\cK} \tr(\tilde{\m V}_k) + \tr(\tilde{\m Q})\\
		\st~~~~ &\quad \tilde{a}_k(\{\tilde{\m V}_k\},\tilde{\m Q}) \geq 0,\fk, \\
		&\quad \tilde{\m B}_m(\{\tilde{\m V}_k\},\tilde{\m Q}) \succeq \m 0 ,\fmm,\\
		&\quad \tilde{\m V}_k \succeq \m 0  ,\fk, 
	\end{aligned}
	\tag{R}
\end{equation}
where $\tilde{\v h}_k = [h_{k,2}, h_{k,3}, \d, h_{k,M}]^\hermitian$ for all $k\in\cK$, $\tilde{\cM} = \{2,3,\d,M\}$, and 
\begin{equation*}
	\begin{aligned}
		\tilde{a}_k(\{\tilde{\m V}_k\},\tilde{\m Q}) & = \frac{1}{\bar \gamma_k} \tilde{\v h}_k^\hermitian \tilde{\m V}_k \tilde{\v h}_k - \sum_{j\neq k} \tilde{\v h}_k^\hermitian \tilde{\m V}_j \tilde{\v h}_k - \tilde{\v h}_k^\hermitian \tilde{\m Q} \tilde{\v h}_k - \sigma_k^2  \\ 
		\tilde{\m B}_m(\{\tilde{\m V}_k\},\tilde{\m Q}) & = 2^{\bar C_m} \begin{bmatrix}
			\m 0 & \m 0  \\
			\m 0 & \tilde{\m Q}^{(m:M, m:M)}
		\end{bmatrix} \\
		& \pushright{-~\left(\sum_{k\in\cK} \tilde{\m V}_k^{(m,m)}+ \tilde{\m Q}^{(m,m)}\right) \m{E}_m. } \\ 
	\end{aligned}
\end{equation*}

Without loss of generality, we assume that any optimal solution $\tilde{\m Q}$ to problem~(R) is positive definite. 
Otherwise, instead of problems (R) and \eqref{SDR}, we can consider problems~(R1) and (R2) obtained from the following relay removal procedure: 
\begin{itemize}
    \item \textbf{(R1)}: Keep removing relays from problem \eqref{SDR} until any optimal solution $\m Q$ to its reduced problem is positive definite. 
    Notice that when there is only one relay, any optimal $\m Q$ (which is a scalar in this case) must be greater than zero. 
    Otherwise, all the relays would be turned off, and this will violate the SINR constraints. 
    Hence, the above relay removal procedure will terminate with a positive number of relays left. 
    Denote the corresponding reduced problem by (R1). 
    \item \textbf{(R2)}: Select one previously removed relay and add it back into problem (R1). 
    Denote the corresponding problem by (R2). 
\end{itemize}

In the following, we formally show that the singular case of the optimal $\m Q$ of problem~\eqref{SDR}, i.e., the optimal values of problems (R) and \eqref{SDR} are the same, has probability zero of occurring. 

\begin{proposition}
\label{thm:PD}
Suppose that any optimal solution $\tilde{\m Q}$ to problem~(R) is positive definite. 
Let $\mathcal A$ be the event that the optimal values of problems (R) and \eqref{SDR} are the same. 
If the channels $\{h_{k,m}\}$ are randomly generated such that $\{h_{k,1} \mid k \in \cK\}$ are statistically independent of $\{h_{k,m} \mid k\in\cK, m\in \tilde{\cM}\}$, then the event $\mathcal A$ has probability zero of occurring. 
\end{proposition}
\begin{proof}
In the following, we first show that the optimal solution of problem~(R) is unique. Then we construct an optimal solution of problem \eqref{SDR} based on the unique solution of problem~(R), derive a condition by plugging the constructed solution into the KKT conditions of problem \eqref{SDR}, and show the derived condition has probability zero of occurring. 

First, given that any optimal solution $\tilde{\m Q}$ to problem~(R) is positive definite, we show that its optimal solution $(\{\tilde{\m V}_k^\star\}, \tilde{\m Q}^\star)$ and its optimal Lagrange multipliers associated with the SINR constraints, denoted by $\tilde{\v \beta}^\star = [\tilde{\beta}_1^\star,\tilde{\beta}_2^\star,\d,\tilde{\beta}_K^\star]^\transpose$, are unique. 
Since any optimal solution $\tilde{\m Q}$ to problem~(R) is positive definite, the KKT conditions of problem~(R) are equivalent to the enhanced KKT conditions of problem~(R) (which are similar to Eqs. \eqsta). 
Applying the same derivation in Sections \ref{ss:dual} and \ref{ss:primal} to problem~(R), we can simplify the enhanced KKT conditions into finding counterparts of Eqs.~\eqref{equ:dual_fix_point}, \eqref{equ:pri_beamforming_direction}, and \eqref{equ:primal_fix_point} over $(\tilde{\v \beta}^\star, \tilde{\v v}_k^\star, \tilde{\v p}^\star)$. 
Furthermore, one can recover the optimal $\{\tilde{\m V}_k^\star\}$ and the optimal $\tilde{\m Q}^\star$ by the procedure described in Section~\ref{ss:primal}. 
The desired uniqueness comes from the uniqueness of fixed points of the SI mappings $I(\cdot)$ and $J(\cdot)$.

Let $(\{\tilde{\m V}_k^\star\}, \tilde{\m Q}^\star)$ be the optimal solution to problem~(R) with $\tilde{\v \beta}^\star$ being its optimal Lagrange multiplier. 
Now we construct an optimal solution of problem~\eqref{SDR} based on the solution of problem~(R).
Since the optimal value of problem \eqref{SDR} is equal to that of problem~(R), $(\{\m V_k^\star := \diag(0, \tilde{\m V}_k^{\star})\}, \m Q^\star := \diag(0, \tilde{\m Q}^\star))$ 
is an optimal solution to problem \eqref{SDR}. 
Hence, there exist $(\v \beta^\star, \{\m \Lambda_m^\star\})$ such that the KKT conditions of problem \eqref{SDR} hold. 
Plugging the constructed $(\{\v V_k^\star\}, \m Q^\star)$ into the KKT conditions of problem~\eqref{SDR} and comparing the KKT conditions of problems~(R) and \eqref{SDR}, we find that $(\{\tilde{\m V}_k^\star\}, \tilde{\m Q}^\star)$ satisfies the KKT conditions of problem~(R) with $\v \beta^\star$ being its optimal Lagrange multiplier. 
By the uniqueness of the optimal Lagrange multiplier of problem~(R), we get $\v \beta^\star = \tilde{\v \beta}^\star$. 

Next, let us focus on the complementary slackness condition, i.e., Eq. \eqref{equ:slack_1}, in the KKT conditions of problem~\eqref{SDR}. 
Plugging $\v \beta^\star = \tilde{\v \beta}^\star$ into Eq. \eqref{equ:slack_1} gives that
\begin{equation}
    \label{equ:start}
    \begin{aligned}
        \tr \Big( \m V_k \big( \m C_k(\tilde{\v \beta}^\star, \{\m \Lambda_m^\star\})  - \tilde{\beta}_k^\star/\bar \gamma_k \v h_k \v h_k^\hermitian\big) \Big) =0,\fk. 
    \end{aligned}
\end{equation}
From the derivation in Section \ref{ss:primal}, we have $\tilde{\m V}_k^\star = \tilde{p}_k^\star \tilde{\v v}_k^\star \left(\tilde{\v v}_k^\star \right)^\hermitian$ for all $k \in \cK$, where $\tilde{p}_k^\star$ is the $k$-th element of $\tilde{\v p}^\star$. 
Using this, we can further rewrite Eq.~\eqref{equ:start} as 
\begin{equation}\label{equ:1}
    \left(\m C_k(\tilde{\v \beta}^\star, \{\m \Lambda_m^\star\}) - \tilde{\beta}_k^\star /\bar \gamma_k \v h_k \v h_k^\hermitian\right) 
    \begin{bmatrix}
        0 \\
        \tilde{\v v}_k^\star
    \end{bmatrix} = \v 0, \fk.
\end{equation}
Focusing on the first equation in Eq.~\eqref{equ:1}, we have
$$
\left(\m C_k(\tilde{\v \beta}^\star, \{\m \Lambda_m^\star\})\right)^{(0, 1:M)}\tilde{\v v}_k^\star = \tilde{\beta}_k^\star/{\bar \gamma_k} \tilde{\v h}_k^\hermitian \tilde{\v v}_k^\star h_{k,1}^\hermitian,
$$
which can be further rewritten as 
\begin{equation}\label{equ:linear_equ}
    -  \tilde{\beta}_k^\star/{\bar \gamma_k} \tilde{\v h}_k^\hermitian \tilde{\v v}_k^\star h_{k,1}^\hermitian + \sum_{j \neq k} \beta_j^\star \tilde{\v h}_j^\hermitian \tilde{\v v}_k^\star h_{j,1}^\hermitian = 0, \fk. 
\end{equation}
Notice that \eqref{equ:linear_equ} is a linear equation of $\{h_{k,1} \mid k \in \cK\}$ and the coefficient $\tilde{\beta}_k^\star/{\bar \gamma_k} \tilde{\v h}_k^\hermitian \tilde{\v v}_k^\star$ is nonzero for all $k\in\cK$. 
Furthermore, the coefficients are solely functions of $\{h_{k,m} \mid k\in\cK, m \in\tilde{\cM}\}$ not including $\{h_{k,1} \mid k\in\cK\}$. 
Let $\mathcal B$ be the event that Eq. \eqref{equ:linear_equ} is satisfied. 
Since $\{h_{k,1} \mid k \in \cK\}$ are randomly generated and independent of $\{h_{k,m} \mid k\in\cK, m \in\tilde{\cM}\}$, the event $\mathcal B$ has probability zero of occurring. 

From the above discussion, we know that $\mathcal A \subseteq \mathcal B$, and the event $\mathcal B$ has probability zero of occurring.  
Hence, the event $\mathcal A$ has probability zero of occurring.
The desired conclusion holds. 

\end{proof}

\section{Details on Solving Eqs.~\eqref{equ:dual_const2} and \eqref{equ:dual_var2} for \texorpdfstring{$ \left\{ \m{\Lambda}_m \right\} $}{\{Lambda\_m\}}}
\label{apd:solve_Lambda}
In this part, we shall provide more details on solving Eqs.~\eqref{equ:dual_const2} and \eqref{equ:dual_var2} for $\left\{ \m \Lambda_m \right\}$. 
For any positive integer $n$ and any $ \eta > 1 $, define the mapping $ S_\eta: \mathcal{S}_{++}^n \rightarrow \mathcal{S}_{++}^{n-1} $ by 
\begin{equation}
	S_\eta(\m{\Gamma}) = \m{\Gamma}^{(2:n, 2:n)} - \frac{\m{\Gamma}^{(2:n,1)} \m{\Gamma}^{(1,2:n)}}{\frac{\eta}{\eta -1}\m{\Gamma}^{(1,1)}}. 
	\label{equ:def_S}
\end{equation}
Let $ \m{\Gamma}_1 = \m{\Gamma}(\v \beta) $, where $\m \Gamma(\cdot)$ is defined in \eqref{equ:def_Gamma}. 
Then by the described procedure in Section \ref{sss:solvingforLm},  $ \m{\Gamma}_{m+1} $ is given by 
\begin{equation}
	\begin{aligned} 
		\mathbf{\Gamma}_{m+1} 
		&= \left[ \begin{matrix} \mathbf{0} &\mathbf{0} \\ 
			\mathbf{0} & S_{2^{\bar C_m}}\left(\m{\Gamma}_m^{(m:M, m:M)}\right) 
		\end{matrix}\right]. 
	\end{aligned}
	\label{equ:solve_L1}
\end{equation}
Once $\left\{ \m \Gamma_m \right\}$ are computed, we can obtain $\Lm$ based on \eqref{equ:Gamma2Lambdam}, and in particular,
\begin{equation}
	\m{\Lambda}_m^{(m,m)} = \frac{1}{2^{\bar C_m}-1} \m{\Gamma}_m^{(m,m)}. 
	\label{equ:solve_L2}
\end{equation}
Based on the recursive formula \eqref{equ:solve_L1} and \eqref{equ:solve_L2}, we have that, for any $m \in \cM$, 
	\begin{equation}
		\m{\Lambda}_m^{(m,m)}\left(\v \beta\right)
		= \frac{1}{2^{\bar C_m} -1} \left( \left( T_m  \circ \m{\Gamma} \right) (\v \beta) \right)^{(1,1)},
		\label{equ:solve_L3}
\end{equation}
where we define
\begin{equation}
	T_m = 
	\begin{cases}
		S_{2^{\bar C_{m-1}}} \circ \cdots \circ S_{2^{\bar C_1}}, & \text{if~} m = 2, 3, \d, M, \\
		\text{identity mapping}, & \text{if~} m=1. 
	\end{cases}
	\label{equ:def_Tm}
\end{equation} 
As will be seen later, \eqref{equ:solve_L3} is useful for the analysis of the properties of $\left\{ \m \Lambda_{m}^{(m,m)}(\cdot) \right\}$ in Appendix~\ref{ss:useful_properties_Lambda}.

\section{Useful Properties of \texorpdfstring{$\left\{ \m \Lambda_m^{(m,m)}(\cdot) \right\}$}{Lambda\_m(m,m)()} and \texorpdfstring{$\m Q(\cdot)$}{Q()}}
\label{apd:useful_properties}

In this part,  we shall derive some useful properties of $\left\{ \m \Lambda_m^{(m,m)}(\cdot) \right\}$ and $\m Q(\cdot)$, which play central roles in showing that $I(\cdot)$ in \eqref{equ:dual_fix_point} and $J(\cdot)$ in \eqref{equ:primal_fix_point} are SI mappings \cite{yates1995FrameworkUplinkPower}, respectively.

Consider the cone $\mathcal C$, which can either be the nonnegative orthant $\mathbb{R}_{+}^n$ or the positive semidefinite matrix cone $\mathcal{S}_{+}^n$. 
The order relationship $\succeq_{\mathcal C}$ ($\succ_{\mathcal C}$) denotes $\geq$ ($>$) if $\mathcal C = \mathbb{R}_{+}^n$ and denotes $\succeq$ ($\succ$) if $\mathcal C = \mathcal{S}_{+}^n$. 
Given two cones $\mathcal{C}_1$ and $\mathcal{C}_2$, a mapping $f: \mathcal C_1 \rightarrow \mathcal C_2$ is called 
\begin{itemize}
	\item [(i)] concave if $f(\lambda \m A_1 + (1-\lambda) \m A_2) \succeq_{\mathcal C_2} \lambda f(\m A_1) + (1-\lambda) f(\m A_2)$ for all $\m A_1, \m A_2 \succ_{\mathcal C_1} \m 0$ and $\lambda \in [0, 1]$; 
	\item [(ii)] homogeneous if $f(\lambda \m A) = \lambda f(\m A)$ for all $\m A \succ_{\mathcal C_1} \m 0$ and $\lambda > 0$; 
	\item [(iii)] nonnegative if $f(\m A) \succeq_{\mathcal C_2} \m 0$ for all $\m A \succeq_{\mathcal C_1} \m 0$; 
	\item [(iv)] strictly nonnegative if $f(\m A) \succ_{\mathcal C_2} \m 0$ for all $\m A \succ_{\mathcal C_1} \m 0$; 
	\item [(v)] monotonic if $f(\m A_1) \succeq_{\mathcal C_2} f(\m A_2)$ for all $\m A_1, \m A_2 \succeq_{\mathcal C_1} \m 0$ with $\m A_1 \succeq_{\mathcal C_1} \m A_2$; and
	\item [(vi)] strictly monotonic if $f(\m A_1) \succ_{\mathcal C_2} f(\m A_2)$ for all $\m A_1, \m A_2 \succ_{\mathcal C_1} \m 0$ with $\m A_1 \succ_{\mathcal C_1} \m A_2$. 
\end{itemize}
Note that the above definition of the nonnegativity and the monotonicity coincides with the definition given in Appendix~\ref{apd:alg_prop} when $f$ maps $\mathbb{R}_+^K$ to $\mathbb{R}_+^K$.

\subsection{Useful Properties of \texorpdfstring{$\left\{ \m \Lambda_m^{(m,m)}(\cdot) \right\}$}{Lambda\_m(m,m)()}}
\label{ss:useful_properties_Lambda}
We study the properties of the obtained components $ \left\{ \m{\Lambda}_m^{(m,m)}(\cdot)\right\} $ in \eqref{equ:solve_L3}. 
As shown in \eqref{equ:solve_L3}, $\m \Lambda_m^{(m,m)}(\cdot)$ is a composition of $\left\{ S_{2^{\bar C_m}}(\cdot) \right\}$ and $\m \Gamma(\cdot)$ in \eqref{equ:def_Gamma}. 
Hence, to study the properties of $\m \Lambda_m^{(m,m)}(\cdot)$, we need to first study the properties of $ S_\eta(\cdot) $ in \eqref{equ:def_S}. 

For any $ \eta > 1$, $S_\eta(\cdot)$ in \eqref{equ:def_S} can be equivalently rewritten as 
\begin{equation*}\small
	\begin{aligned}
		S_\eta(\m{\Gamma}) &= \frac{\eta-1}{\eta} \left(\m{\Gamma}^{(2:n, 2:n)} - \frac{\m{\Gamma}^{(2:n,1)} \m{\Gamma}^{(1,2:n)}}{\m{\Gamma}^{(1,1)}} \right) + \frac{1}{\eta} \m{\Gamma}^{(2:n, 2:n)}, 
	\end{aligned}
\end{equation*}
which is a convex combination of a Schur complement and a linear part of $\m \Gamma$. 
The concavity and homogeneity of $ S_\eta(\cdot) $ come from the concavity and homogeneity of the Schur complement \cite[Exercise 3.58]{boyd2004ConvexOptimization}. 
Moreover, $ S_\eta(\cdot) $ is also strictly nonnegative. 
Finally, we have the following lemma on the monotonicity of $S_\eta(\cdot)$. 
\begin{lemma}
	\label{lem:mono}
	For a concave and homogeneous mapping $f: \mathcal C_1 \rightarrow \mathcal C_2$, the following holds. 
	\begin{itemize}
		\item [(i)] If $f$ is nonnegative, then it is monotonic. In particular, this conclusion holds when $f$ is linear and nonnegative. 
		\item [(ii)] If $f$ is strictly nonnegative, then it is strictly monotonic. 
	\end{itemize}
\end{lemma}
\begin{proof}
	For any $\m A_1, \m A_2 \succeq_{\mathcal C_1} \m 0$ such that $\m A_1 \succeq_{\mathcal C_1} \m A_2$, we have
	\begin{equation*}
	f(\m A_1) \succeq_{\mathcal C_2} f(\m A_2) + f(\m A_1 - \m A_2) \succeq_{\mathcal C_2} f(\m A_2),  
	\end{equation*}
	where the first inequality follows from the concavity and homogeneity of $f$, and the second inequality follows from the nonnegativity of $f$. 
	When $f$ is strictly nonnegative, the proof can be done analogously by replacing $\succeq$ in the conditions and the second inequality with $\succ$. 
\end{proof}

Next, we shall study the properties of the composite mapping $ T_m $ in \eqref{equ:def_Tm}. 
\begin{lemma}
	\label{lem:comp}
	If $f$ and $g$ are two concave, homogeneous, and strictly nonnegative mappings, so does their composition $f \circ g$. 
\end{lemma}
\begin{proof}
	From Lemma \ref{lem:mono}, we know that $f$ is monotonically increasing. 
	Since $g$ is concave, and $f$ is concave and nondecreasing, their composition $f \circ g$ is concave \cite[Chapter 3]{boyd2004ConvexOptimization}. 
	Besides, the homogeneity and strict nonnegativity of the composition $f \circ g$ is obvious. 
\end{proof}

The concavity, the homogeneity, and the strict nonnegativity of $ T_m $ in \eqref{equ:def_Tm} come from the following induction argument. 
First, $ T_2 = S_{2^{\bar C_1}} $ is concave, homogeneous, and strictly nonnegative. 
Lemma \ref{lem:comp} shows that if $ T_{m-1} $ and $ S_{2^{\bar C_{m - 1}}} $ have these properties, then $ T_{m} = S_{2^{\bar C_{m - 1}}} \circ T_{m-1} $ also has these properties. 
Hence, for any $ m \in \cM $, $ T_m $ is concave, homogeneous, and strictly nonnegative. 
Furthermore, the strict monotonicity of $T_m(\cdot)$ follows from Lemma \ref{lem:mono}. 

Now, we present the nice properties of $ \left\{ \m{\Lambda}_m^{(m,m)}(\cdot) \right\} $ in the following Lemma \ref{lem:lmmm}. 

\begin{lemma}
	\label{lem:lmmm}
	For any $ m \in \cM $, $\m{\Lambda}_m^{(m,m)}(\cdot)$ in \eqref{equ:solve_L3} is positive, strictly subhomogeneous, and monotonic. 
\end{lemma}
\begin{proof}
We prove these properties of $\m \Lambda_m^{(m,m)}(\cdot)$ one by one. 

\textbf{Positivity}: For any $ \v \beta \in \mathbb{R}_+^K $, we have $ \m{\Gamma}(\v \beta)\succeq \m{I} \succ \m 0  $ by \eqref{equ:def_Gamma}. 
	Therefore, the strict nonnegativity of $ T_m(\cdot) $ and \eqref{equ:solve_L3} yields 
	\begin{equation*}
		\begin{aligned}
			\m{\Lambda}_m^{(m,m)}\left(\v \beta\right)
			&= \frac{1}{2^{\bar C_m}-1} \Big(T_m \big(\m{\Gamma}(\v \beta)\big)\Big)^{(1,1)} > 0. 
		\end{aligned}
	\end{equation*}

\textbf{Strict subhomogeneity}: For any $ \alpha > 1 $ and $\v \beta \in \mathbb{R}_{+} \backslash \left\{ \v 0 \right\}$, it is obvious from \eqref{equ:def_Gamma} that $ \m{\Gamma}\left(\alpha \v \beta\right)\prec \alpha\m{\Gamma}(\v \beta) $. The strict monotonicity and homogeneity of $ T_m(\cdot) $ together with \eqref{equ:solve_L3} give 
	\begin{equation*}
		\begin{aligned}
			\m{\Lambda}_m^{(m,m)}\left(\alpha\v \beta\right)
			&= \frac{1}{2^{\bar C_m}-1} \Big(T_m \big(\m{\Gamma}(\alpha\v \beta)\big)\Big)^{(1,1)}\\
			&< \frac{1}{2^{\bar C_m}-1} \Big(T_m \big(\alpha\m{\Gamma}(\v \beta)\big)\Big)^{(1,1)}\\
			&= \frac{\alpha}{2^{\bar C_m}-1} \Big(T_m \big(\m{\Gamma}(\v \beta)\big)\Big)^{(1,1)}\\
			&= \alpha \m{\Lambda}_m^{(m,m)}\left(\v \beta\right). 
		\end{aligned}			
	\end{equation*}

\textbf{Monotonicity}: For any $ \v \beta_1, \v \beta_2 \in \mathbb{R}_+^K $ with $\v \beta_1 \geq \v \beta_2$, we have $ \m{\Gamma}\left(\v \beta_1\right)\succeq \m{\Gamma}\left(\v \beta_2\right) $ by \eqref{equ:def_Gamma}. 
	By the monotonicity of $ T_m(\cdot) $ and \eqref{equ:solve_L3}, we have 
	\begin{equation*}
		\begin{aligned}
			\m{\Lambda}_m^{(m,m)}\left(\v \beta_2\right)
			&= \frac{1}{2^{\bar C_m}-1} \Big(T_m \big(\m{\Gamma}(\v \beta_2)\big)\Big)^{(1,1)}\\
			&\leq \frac{1}{2^{\bar C_m}-1} \Big(T_m \big(\m{\Gamma}(\v \beta_1)\big)\Big)^{(1,1)}\\
			&= \m{\Lambda}_m^{(m,m)}\left(\v \beta_1\right). 
		\end{aligned}
	\end{equation*}
\end{proof}

\subsection{Useful Properties of \texorpdfstring{$\m Q(\cdot)$}{Q()}}
\label{ss:useful_properties_Q}
We have the following lemma about the useful properties of $ \m Q(\cdot) $ given by \eqref{equ:QMM}--\eqref{equ:solve_Q3}. 
\begin{lemma}
	\label{lem:Q}
	The mapping $\m Q(\cdot)$ given by \eqref{equ:QMM}--\eqref{equ:solve_Q3} is nonnegative, linear, and monotonic. 
\end{lemma}
\begin{proof}
	It suffices to show the linearity and nonnegativity of $\m Q(\cdot)$ since the monotonicity of $\m Q(\cdot)$ automatically follows from these two properties using Lemma \ref{lem:mono}. In the following, we show these two properties of $\m Q(\cdot)$ by induction.

	It is obvious that $ \m Q^{(M,M)} $ given in \eqref{equ:QMM} is nonnegative and linear. 
	Suppose $ \m Q^{(m+1:M,m+1:M)} $ is nonnegative and linear. 
	First, by \eqref{equ:solve_Q1} and \eqref{equ:solve_Q3}, $\m Q^{(m+1:M, m)}$ and $\m Q^{(m,m)}$ is linear. 
	Combining this with the assumption shows the linearity of $\m Q^{(m:M, m:M)}$. 
	Next, by the elementary properties of the Schur complements \cite[Theorem 1.20]{zhang2005SchurComplementIts}, to show the nonnegativity of $\m Q^{(m:M, m:M)}$, it suffices to check the following two conditions. 
	\begin{itemize}
		\item [(i)] The Schur complement of $\m Q^{(m+1:M,m+1:M)}$ in $ \m Q^{(m:M,m:M)}$ is nonnegative; and 
		\item [(ii)] for any $\v x \in \mathbb{C}^{M-m-2}$ with $ \m Q^{(m+1:M, m+1:M)} \v x = \v 0$, we have $ \m Q^{(m,m+1:M)} \v x = \v 0$. 
	\end{itemize}
	Plugging \eqref{equ:solve_Q1} and \eqref{equ:solve_Q3} into expressions in the above (i) and (ii), we have 
	\begin{equation}
		\label{equ:cond_1}
		\begin{aligned}
			&\m Q^{(m,m)} - \m Q^{(m,m+1:M)} \left(\m Q^{(m+1:M,m+1:M)}\right)^{-1} \m Q^{(m+1:M,m)}\\ 
			&= \frac{1}{2^{\bar C_m} - 1} \Biggl( \frac{\v \lambda_m^{(m+1:M)\hermitian} \m Q^{(m+1:M, m+1:M)} \v \lambda_m^{(m+1:M)}}{|\v \lambda_m^{(m)}|^2} \\
			&\pushright{+ \sum_{k \in \cK} p_k \hat{\m V}_k^{(m,m)} \Biggr) \geq 0}; 
		\end{aligned}
	\end{equation}
	and for any $\v x \in \mathbb{C}^{M-m-2}$ with $ \m Q^{(m+1:M, m+1:M)} \v x = \v 0$, we have 
	\begin{equation}
		\label{equ:cond_2}
		\begin{aligned}
			\m Q^{(m,m+1:M)} \v x = \frac{\v \lambda_{m}^{(m+1:M)\hermitian} \m Q^{(m+1:M, m+1:M)} \v x}{\v \lambda_m^{(m)\hermitian}} = \v 0. 
		\end{aligned}
	\end{equation}
	Hence, we get the nonnegativity of $\m Q^{(m:M,m:M)}$. 
	It follows by induction that $\m Q^{(m:M,m:M)}$ is nonnegative and linear for any $m \in \cM$. 
	Taking $m=1$, we obtain the desired result. 
\end{proof}



\ifCLASSOPTIONcaptionsoff
  \newpage
\fi

\section*{Reference}\footnotesize
\begin{itemize}
    \item[{[2]}] S. P. Boyd and L. Vandenberghe, \emph{Convex Optimization}.
Cambridge University Press, 2004.
    \item[{[48]}] L. Liu, Y.-F. Liu, P. Patil, and W. Yu, ``Uplink-downlink duality between
multiple-access and broadcast channels with compressing relays,'' \emph{IEEE
Trans. Inf. Theory}, pp. 7304--7337, Nov. 2021.
    \item[{[60]}] R. Yates, ``A framework for uplink power control in cellular radio
systems,'' \emph{IEEE J. Sel. Areas Commun.}, vol. 13, no. 7, pp. 1341--1347,
Sep. 1995.
    \item[{[61]}] F. Zhang, Ed., \emph{The Schur Complement and Its Applications}. New York:
Springer, 2005.
\end{itemize}
\end{document}